\documentclass[a4paper,UKenglish, autoref]{lipics-v2019}
\usepackage[utf8]{inputenc}
\usepackage{amsmath,amssymb,dsfont}
\usepackage{mathtools}
\usepackage{thmtools}
\usepackage{graphicx}
\usepackage{float}
\usepackage{xspace}
\usepackage[normalem]{ulem}
\usepackage{algorithm}
\usepackage{algorithmic}

\usepackage{complexity}
\usepackage{subcaption}
\usepackage{hyperref}
\hypersetup{colorlinks = true, citecolor = dark-red, linkcolor = blue, breaklinks = true, urlcolor=dark-blue}
\usepackage[nameinlink,capitalise,noabbrev]{cleveref}
\usepackage{soul}

\usepackage{tikz}
\usepackage{tkz-graph}
\usetikzlibrary{arrows,calc,decorations.pathreplacing,shapes,decorations.pathmorphing}

\usepackage[final]{pdfpages}

\newcommand{\NN}{\mathbb{N}}

\newcommand{\pname}[1]{\textsc{#1}\xspace}

\newcommand{\Problem}[3]{
	\begin{flushleft}
		\fbox{
		\begin{minipage}{.96\textwidth}
			\noindent {\pname{#1}}\\
			{\bf Instance:} #2\\
			{\bf Question:} #3
		\end{minipage}}
		\medskip
	\end{flushleft}
}

\definecolor{color-test}{rgb}{0.60, 0.0, 0.80}
\newcommand{\ig}[1]{\textcolor{red}{$\langle${\sf Ig}: #1$\rangle$}}
\newcommand{\igr}[1]{\textcolor{medium-blue}{$\langle$#1$\rangle$}}

\renewcommand{\b}[1]{\textcolor{orange}{#1}}

\newcommand{\rud}[2]{\textcolor{green}{#1}\ \textcolor{red}{ \sout{#2}}}
\newcommand{\rudC}[1]{\textcolor{green}{\bf $\langle${\sf Rudini}: #1$\rangle$}}

\newcommand{\vini}[1]{\textcolor{orange}{$\langle${\sf Vini}: #1$\rangle$}}
\newcommand{\orange}[1]{\textcolor{orange}{#1}}

\newcommand{\magenta}[1]{\textcolor{magenta}{#1}}

\renewcommand{\b}[1]{#1}
\renewcommand{\magenta}[1]{#1}

\definecolor{dark-red}{rgb}{0.7,0.15,0.15}
\definecolor{dark-blue}{rgb}{0.15,0.15,0.4}
\definecolor{medium-blue}{rgb}{0,0,0.9}
\definecolor{gray}{rgb}{0.5,0.5,0.5}
\definecolor{color-Ig}{rgb}{0.15,0.7,0.15}

\newcommand{\tw}{{\sf tw}\xspace}
\newcommand{\yes}{{\sf yes}\xspace}

\renewcommand{\FPT}{{\sf FPT}\xspace}
\newcommand{\MSO}{${\sf MSO}$\xspace}
\newcommand{\MSOone}{${\sf MSO}_1$\xspace}

\newcommand{\tbp}{t_{P_3}\xspace}
\newcommand{\Ocal}{\mathcal{O}}

\usepackage{cite}
\title{Target~set~selection~with~maximum~activation~time}

\author{Lucas Keiler}{Dept. Computação, Centro de Ciências, Universidade Federal do Ceará,
  Fortaleza - CE, Brazil}{lucas.keiler@hotmail.com}{}{}%TODO mandatory, please use full name; only 1 author per \author macro; first two parameters are mandatory, other parameters can be empty. Please provide at least the name of the affiliation and the country. The full address is optional

\author{Carlos Vinicius Gomes Costa Lima}{Dept. Computação, Centro de Ciências, Universidade Federal do Ceará,
  Fortaleza - CE, Brazil}{gclima@lia.ufc.br}{https://orcid.org/0000-0002-6666-0533}{}%TODO mandatory, please use full name; only 1 author per \author macro; first two parameters are mandatory, other parameters can be empty. Please provide at least the name of the affiliation and the country. The full address is optional

\author{Ana Karolinna Maia}{Dept. Computação, Centro de Ciências, Universidade Federal do Ceará,
  Fortaleza - CE, Brazil}{karolmaia@ufc.br}{https://orcid.org/0000-0002-9027-7948}{}%TODO mandatory, please use full name; only 1 author per \author macro; first two parameters are mandatory, other parameters can be empty. Please provide at least the name of the affiliation and the country. The full address is optional

\author{Rudini Sampaio}{Dept. Computação, Centro de Ciências, Universidade Federal do Ceará,
  Fortaleza - CE, Brazil}{rudini@ufc.br}{https://orcid.org/0000-0001-5889-5183}{}%TODO mandatory, please use full name; only 1 author per \author macro; first two parameters are mandatory, other parameters can be empty. Please provide at least the name of the affiliation and the country. The full address is optional

\author{Ignasi Sau}{LIRMM, Universit\'e de Montpellier, CNRS, Montpellier, France}{ignasi.sau@lirmm.fr}{https://orcid.org/0000-0002-8981-9287}{}
\authorrunning{L.\ Keiler, C.\ V.\ G.\ C.\ Lima, A.\ K.\ Maia, R.\ Sampaio and I.\ Sau}%TODO mandatory. First: Use abbreviated first/middle names. Second (only in severe cases): Use first author plus 'et al.'

\Copyright{Lucas Keiler, Carlos V.\ G.\ C.\ Lima, Ana K.\ Maia, Rudini Sampaio and Ignasi Sau}%TODO mandatory, please use full first names. LIPIcs license is "CC-BY";  http://creativecommons.org/licenses/by/3.0/

%\textcolor{red}{Replace ccsdesc macro with valid one}} %TODO mandatory: Please choose ACM 2012 classifications from https://dl.acm.org/ccs/ccs_flat.cfm

%TODO mandatory; please add comma-separated list of keywords

\funding{Supported by CAPES [88887.143992/2017-00] DAAD Probral and [88881.197438/2018-01] STIC AmSud, CNPq Universal [401519/2016-3], [425297/2016-0] and [437841/2018-9],  FUNCAP [4543945/2016] Pronem, and French projects DEMOGRAPH (ANR-16-CE40-0028) and ESIGMA (ANR-17-CE23-0010).}%optional, to capture a funding statement, which applies to all authors. Please enter author specific funding statements as fifth argument of the \author macro.

\nolinenumbers %uncomment to disable line numbering

\hideLIPIcs  %uncomment to remove references to LIPIcs series (logo, DOI, ...), e.g. when preparing a pre-final version to be uploaded to arXiv or another public repository

\EventEditors{}
\EventNoEds{2}
\EventLongTitle{}
\EventShortTitle{ISAAC 2020}
\EventAcronym{}
\EventYear{}
\EventDate{}
\EventLocation{}
\EventLogo{}
\SeriesVolume{}
\ArticleNo{1}
\begin{document}
%\mainmatter
\maketitle

\renewcommand{\c}[1]{\textcolor{color-test}{#1}}

%TODO mandatory: add short abstract of the document
\begin{abstract}
%In the \TSS model, it is given a graph $G$ and a threshold function $\tau:V\to\NN$ upper-bounded by the vertex degree, as defined in \cite{ningchen}.
A \emph{target set selection model} is a graph $G$ with a threshold function $\tau:V\to\NN$ upper-bounded by the vertex degree.
For a given model, a set~$S_0\subseteq V(G)$ is a \emph{target set} if $V(G)$ can be partitioned into non-empty subsets $S_0,S_1,\dotsc,S_t$ such that, for $i \in \{1, \dotsc, t\}$,
$S_i$ contains exactly every vertex $v$
%outside $S_0\cup\dots\cup S_{i-1}$
having at least $\tau(v)$ neighbors in $S_0\cup\dots\cup S_{i-1}$.
We say that~$t$ is the \emph{activation time} $t_{\tau}(S_0)$ of the target set $S_0$.
The problem of, given such a model, finding  a target set of minimum size has been extensively studied in the literature.
In this article, we investigate its variant, which we call \textsc{TSS-time}, in which the goal is to find a target set~$S_0$ that maximizes $t_{\tau}(S_0)$.
That is, given a graph $G$, a threshold function $\tau$ in $G$, and an integer $k$, the objective of the \textsc{TSS-time} problem is to decide whether $G$ contains a target set~$S_0$ such that $t_{\tau}(S_0)\geq k$.
Let $\tau^* = \max_{v \in V(G)} \tau(v)$.
%It turns out that the complexity of this problem in minor-closed graph classes is determined by the property of having bounded local treewidth.
Our main result is the following dichotomy about the complexity of \pname{TSS-time} when~$G$ belongs to a minor-closed graph class ${\cal C}$: if ${\cal C}$ has bounded local treewidth, the problem is \FPT parameterized by~$k$ and $\tau^{\star}$; otherwise, it is \NP-complete even for fixed $k=4$ and $\tau^{\star}=2$.
We also prove that, with~$\tau^*=2$, the problem is \NP-hard in bipartite graphs for fixed~$k=5$,
and from previous results we observe that \textsc{TSS-time} is \NP-hard in planar graphs and \W[1]-hard parameterized by treewidth.
%\ig{the planar case and the bounded treewidth case are rather ``observations'', since the result follows almost directly from exsisting results}
%\rud{I AM NOT SURE WHAT TO DO HERE. I THINK IT IS NOT A PROBLEM THE WAY IT IS. BUT IF YOU WANT TO REMOVE, NO PROBLEM.}{}
Finally, we present a linear-time algorithm to find a target set $S_0$ in a given tree maximizing $t_{\tau}(S_0)$.
%, for \c{any} \ig{really?} threshold function.

%\b{We also prove that the problem is \FPT in planar graphs parameterized by the activation time.} \ig{para-\NP-hard?}
\end{abstract}

%\keywords{Target set selection, activation time, fixed-parameter tractability, bounded local treewidth, trees, bipartite graphs.}

\iffalse
\fbox{Color coding of Ignasi's comments:}
\begin{itemize}
\item[$\bullet$] I use \ig{this macro} for comments.
\item[$\bullet$] I use \b{blue} for stuff that I added to the text (without comment). I didn't use colors for small typos, commas, notation...
\item[$\bullet$] I use \c{this color} for things that were written, on which I want to comment right after using \ig{...}, or just to replace by the blue stuff inside \igr{...}.
\item[$\bullet$] The macros that I added are in the ``new stuff from Ignasi'', in the preamble of the main tex file.
\end{itemize}

\fbox{Color coding of Rudini's comments:}
\begin{itemize}
\item[$\bullet$] I use \rudC{green with brackets} for comments and \rud{this one to added text}{}.
\end{itemize}

\vspace{.2cm}
\fbox{Color coding of Vinicius' comments:}
\begin{itemize}
\item[$\bullet$] I use \vini{orange with brackets} for comments and \orange{this one to added text}.
\end{itemize}
\fi

\newpage

% ---------------------------------------------------------------
% Sections
\section{Introduction}\label{sec:Intro}

%We consider only undirected graphs without loops nor multiple edges.

In the \emph{target set selection model} (\textsc{TSS} model for short), as formulated by Chen~\cite{ningchen}, it is given an undirected graph $G=(V,E)$ and a \emph{threshold function} $\tau$ in $G$, which is a function $\tau:V(G)\to\NN$ satisfying $1\leq\tau(v)\leq d(v)$ for every vertex~$v$, where $d(v)$ is the degree of $v$.
We say that a set $S_0\subseteq V(G)$ is a \emph{target set} if the vertex set $V(G)$ can be partitioned into non-empty subsets $S_0,S_1,\dotsc,S_t$ such that, for~$i\in\{1,\dotsc,t\}$, $S_i$ contains exactly every vertex $v$ outside $S_0\cup\dots\cup S_{i-1}$ having at least~$\tau(v)$ neighbors in $S_0\cup\dots\cup S_{i-1}$.
We say that $t$ is the \emph{activation time} $t_{\tau}(S_0)$ of the target set $S_0$,
since this represents an activation process starting with $S_0$: initially all vertices in the target set $S_0$ become \emph{active}, the other vertices are \emph{inactive}, and active vertices remain active forever (that is, it is an irreversible and monotone process). At any step of the process, each inactive vertex gets activated if the number of its active neighbors is at least its threshold. The process is synchronous, that is, all inactive vertices update their status at the same time in each step of the process.

The \pname{Influence Maximization} problem, that consists in finding a subset of vertices of size $k$ that maximizes the expected number of vertices activated by the process described above in a given graph, was first studied by Kempe et al. \cite{kempe2003,KempeKT15} with thresholds randomly chosen from a given range. The \textsc{TSS} model defined above (with deterministic thresholds) was introduced in \cite{ningchen}, where the considered problem was to find a target set $S_0$ of minimum size. \b{Let us call this problem \textsc{TSS-size}.}
 %that allows all the vertices of a graph to became active in a finite number of steps, the classical \pname{TSS} problem, was studied.
 Since then, a number of articles investigated the \textsc{TSS-size} problem~\cite{nichterlein2013, cordasco2018, chopin2014, chiang2013, ACKERMAN2010, BENZWI2011, bazgan2014,EHARD2019}.

For a \textsc{TSS} model, there may exist different targets sets (of minimum size or not), yielding different activation times.
%partitions $S_0,S_1,\dotsc,S_t$ of the vertex set of a graph~$G$, each one defining more or less parts, and thus, different activation times. %In this paper, we investigate the problem of maximum activation time (\pname{TSS-time} for short).
\b{Motivated by a recent line of research arising from a question of Bollob\'as on extremal properties of a closely related model discussed below (see the introduction of~\cite{Przykucki} and~\cite{wg2014-tcs} for a detailed discussion),} we define the following parameter: the \emph{maximum activation time} $t_{\tau}(G)$ \b{of a \textsc{TSS} model $(G, \tau)$} is the maximum~$t_{\tau}(S_0)$ among all target sets $S_0$ of $G$.
We consider the complexity of the decision version of the problem of determining $t_{\tau}(G)$, defined as follows.

\Problem{Target Set Selection-Time (TSS-time)}
{A graph $G=(V,E)$, a threshold function $\tau:V(G)\to\NN$,
%( $1\leq\tau(v)\leq d(v)$ for every $v\in V(G)$)
 and a positive integer $k$.}
{Is $t_{\tau}(G)\geq k$?}

For an instance $(G,\tau,k)$ of the \textsc{TSS-time} problem, we let $\tau^*:= \max_{v \in V(G)}\tau(v)$. To the best of our knowledge, the above problem had not been considered before \b{(we discuss related work in the next paragraph)}.
%we are the first to investigate the maximum time for a target set to activate all the vertices of a \pname{TSS} model.
%\ig{at some point we have to say that we may assume that the input graph $G$ of the \textsc{TSS-time} problem is connected, as otherwise we may solve the problem independently in each connected component.} YES, BUT THIS IS NOT REALLY NECESSARY.
Clearly, we may assume $G$ as connected, since otherwise we may solve the problem independently in each connected component. Intuition suggests that the maximum time of activation processes might be obtained with minimum target sets, but this is not true in general. For example, \cref{fig:Fig1} depicts a tree $T$ formed by a root vertex $v$ together with~$k+1$ disjoint induced paths on~$t$ vertices, $v$ being adjacent to an endvertex of each path, for positive integers $t$ and $k>1$. The thresholds are in red and the target sets are marked in gray. \cref{fig:po(T)=1} represents the (unique) minimum target set $S$ of~$T$ of size one (containing the root $v$) with activation time $t$ (achieved at the leaves of $T$). \cref{fig:t(T)=2t} shows a target set $S'$ of size $k$ given by any $k$ leaves of $T$, with activation time $2t$ (achieved at the leaf labeled $u_{t,k+1}$).
%We can see that~$v$ gets activated at time $t$, after all vertices of every path containing a leaf in~$S'$ have been activated. After that, the remaining path is completely activated at time~$2t$. In order to prove that the maximum time of $T$ is $2t$, it is easy to see that~$v$ must be activated in at most~$t$ time steps, which is attained by $S'$, and, once $v$ gets activated, all the remaining vertices are activated in at most $t$ time steps.
Thus, the difference between the cardinalities of a minimum target set and of a target set achieving the maximum activation time can be arbitrarily large.

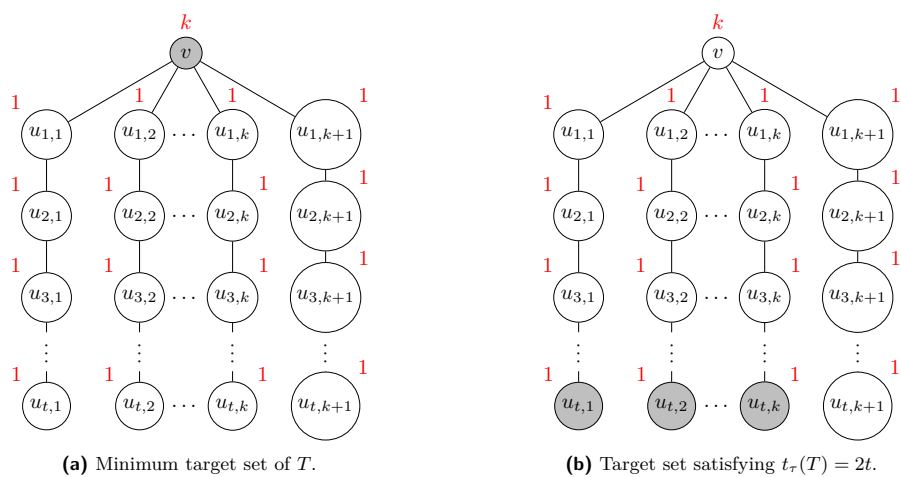
\begin{figure}[t!]
\centering\scalebox{0.8}{
    \begin{subfigure}[b]{0.6\textwidth}
	    \centering
		\begin{tikzpicture}[scale=0.9,
		level/.style={sibling distance=17mm/#1},
		level distance = 1.5cm,
		level 4/.style = {level distance=1cm},
        level 5/.style = {level distance=1cm}
        ]
		
		    \tikzstyle{vertex}=[draw,circle,fill=white,minimum size=15pt,inner sep=2pt]
		
		    \node[vertex] (v) at (0,0) [fill=black!25, label=above:\textcolor{red}{$k$}] {$v$}
				child {node[vertex] (u11) [label=above left:\textcolor{red}{$1$}] {$u_{1,1}$}
				    child {node[vertex] (u21) [label=above left:\textcolor{red}{$1$}] {$u_{2,1}$}
				        child {node[vertex] (u31) [label=above left:\textcolor{red}{$1$}] {$u_{3,1}$}
				            child {node (ret1) [midway,rotate=90] {$\dots$}
				                child {node[vertex] (ut1) [label=above left:\textcolor{red}{$1$}] {$u_{t,1}$}}
				            }
				        }
				    }
				}
				child {node[vertex] (u12) [label=above:\textcolor{red}{$1$}] {$u_{1,2}$}
				    child {node[vertex] (u22) [label=above left:\textcolor{red}{$1$}] {$u_{2,2}$}
				        child {node[vertex] (u32) [label=above left:\textcolor{red}{$1$}] {$u_{3,2}$}
				            child {node (ret2) [midway,rotate=90] {$\dots$}
				                child {node[vertex] (ut2) [label=above left:\textcolor{red}{$1$}] {$u_{t,2}$}}
				            }
				        }
				    }
				}
				child {node[vertex] (u1k) [label=above:\textcolor{red}{$1$}] {$u_{1,k}$}
				    child {node[vertex] (u2k) [label=above right:\textcolor{red}{$1$}] {$u_{2,k}$}
				        child {node[vertex] (u3k) [label=above right:\textcolor{red}{$1$}] {$u_{3,k}$}
				            child {node (retk) [midway,rotate=90] {$\dots$}
				                child {node[vertex] (utk) [label=above right:\textcolor{red}{$1$}] {$u_{t,k}$}}
				            }
				        }
				    }
				}
				child {node[vertex] (u1k+1) [label=above right:\textcolor{red}{$1$}] {$u_{1,k+1}$}
				    child {node[vertex] (u2k+1) [label=above right:\textcolor{red}{$1$}] {$u_{2,k+1}$}
				        child {node[vertex] (u3k+1) [label=above right:\textcolor{red}{$1$}] {$u_{3,k+1}$}
				            child {node (retk+1) [midway,rotate=90] {$\dots$}
				                child {node[vertex] (utk+1) [label=above right:\textcolor{red}{$1$}] {$u_{t,k+1}$}}
				            }
				        }
				    }
				}
			;
			
			\path (u12) -- (u1k) node (r1) [midway] {$\dots$};
			\path (u22) -- (u2k) node (r2) [midway] {$\dots$};
			\path (u32) -- (u3k) node (r3) [midway] {$\dots$};
			\path (ut2) -- (utk) node (rt) [midway] {$\dots$};
		
		\end{tikzpicture}
	    \caption{Minimum target set of $T$.}
		\label{fig:po(T)=1}
	\end{subfigure}
	~
	\begin{subfigure}[b]{0.6\textwidth}
	    \centering
		\begin{tikzpicture}[scale=0.9,
		level/.style={sibling distance=17mm/#1},
		level distance = 1.5cm,
		level 4/.style = {level distance=1cm},
        level 5/.style = {level distance=1cm}
        ]
		
		    \tikzstyle{vertex}=[draw,circle,fill=white,minimum size=15pt,inner sep=2pt]
		
		    \node[vertex] (v) at (0,0) [label=above:\textcolor{red}{$k$}] {$v$}
				child {node[vertex] (u11) [label=above left:\textcolor{red}{$1$}] {$u_{1,1}$}
				    child {node[vertex] (u21) [label=above left:\textcolor{red}{$1$}] {$u_{2,1}$}
				        child {node[vertex] (u31) [label=above left:\textcolor{red}{$1$}] {$u_{3,1}$}
				            child {node (ret1) [midway,rotate=90] {$\dots$}
				                child {node[vertex] (ut1) [fill=black!25, label=above left:\textcolor{red}{$1$}] {$u_{t,1}$}}
				            }
				        }
				    }
				}
				child {node[vertex] (u12) [label=above:\textcolor{red}{$1$}] {$u_{1,2}$}
				    child {node[vertex] (u22) [label=above left:\textcolor{red}{$1$}] {$u_{2,2}$}
				        child {node[vertex] (u32) [label=above left:\textcolor{red}{$1$}] {$u_{3,2}$}
				            child {node (ret2) [midway,rotate=90] {$\dots$}
				                child {node[vertex] (ut2) [fill=black!25, label=above left:\textcolor{red}{$1$}] {$u_{t,2}$}}
				            }
				        }
				    }
				}
				child {node[vertex] (u1k) [label=above:\textcolor{red}{$1$}] {$u_{1,k}$}
				    child {node[vertex] (u2k) [label=above right:\textcolor{red}{$1$}] {$u_{2,k}$}
				        child {node[vertex] (u3k) [label=above right:\textcolor{red}{$1$}] {$u_{3,k}$}
				            child {node (retk) [midway,rotate=90] {$\dots$}
				                child {node[vertex] (utk) [fill=black!25, label=above right:\textcolor{red}{$1$}] {$u_{t,k}$}}
				            }
				        }
				    }
				}
				child {node[vertex] (u1k+1) [label=above right:\textcolor{red}{$1$}] {$u_{1,k+1}$}
				    child {node[vertex] (u2k+1) [label=above right:\textcolor{red}{$1$}] {$u_{2,k+1}$}
				        child {node[vertex] (u3k+1) [label=above right:\textcolor{red}{$1$}] {$u_{3,k+1}$}
				            child {node (retk+1) [midway,rotate=90] {$\dots$}
				                child {node[vertex] (utk+1) [label=above right:\textcolor{red}{$1$}] {$u_{t,k+1}$}}
				            }
				        }
				    }
				}
			;
			
			\path (u12) -- (u1k) node (r1) [midway] {$\dots$};
			\path (u22) -- (u2k) node (r2) [midway] {$\dots$};
			\path (u32) -- (u3k) node (r3) [midway] {$\dots$};
			\path (ut2) -- (utk) node (rt) [midway] {$\dots$};
		
		\end{tikzpicture}
	    \caption{Target set satisfying $t_{\tau}(T)=2t$.}
		\label{fig:t(T)=2t}
	\end{subfigure}}

    \caption{A tree $T$ with (unique) minimum target set $S$ of size one and $t_{\tau}(S) = t$, and such that $t_{\tau}(G) = 2t$, for any positive integer~$t$. The thresholds are in red, while the vertices of~$S$ and $S'$ are marked in gray in \cref{fig:po(T)=1} and \cref{fig:t(T)=2t}, respectively.
    %\ig{for the labels of the vertices  $u_{ij}$ in the figure, I prefer $u_{i,j}$ (with a comma). In fact, I wonder whether we need to include the labels of vertices $u_{i,j}$ in the figure} \edicao{DONE}{}
    }
    \label{fig:Fig1}
\end{figure} 

There are several recent articles in the literature  dealing with problems similar to \pname{TSS-time}, but considering different models \magenta{or different activation processes}. For example, the \emph{$r$-neighbor bootstrap percolation model}~\cite{riedl,morris,balogh4,balogh3,neural2,balogh2,balogh,chalupa,holroyd} is almost equivalent to the \pname{TSS} model with all thresholds being equal to $r$ but it accepts thresholds greater than the degree of a vertex.
Motivated by this, we define a \emph{generalized threshold function} in a graph $G$ as any function $\tau':V(G)\to\NN$. Thus, a threshold function $\tau$ is a generalized threshold function satisfying $1\leq\tau(v)\leq d(v)$ for any vertex $v$ of~$G$.
Let the \emph{generalized \pname{TSS} model} be the analogous of \pname{TSS} model with generalized threshold functions, instead of threshold functions.
Hence, for an integer $r>0$, the $r$-neighbor bootstrap percolation model is equivalent to the generalized \pname{TSS} model with all thresholds equal to $r$.
The $r$-\textsc{Neighbor Bootstrap Percolation-time} and \textsc{Generalized Target Set Seletion-time} \b{(\pname{GTSS-time} for short)} problems are defined analogously to \textsc{TSS-time} for the corresponding models. Observe that, for those problems, vertices whose threshold is greater than its degree must be in any target set (activation time 0) and vertices with threshold 0 are always activated at time at most 1.

%\c{That is, the 2-\textsc{Neighbor Bootstrap Percolation} model accepts instances with degree one vertices, but \pname{TSS} with all thresholds being equal to two does~not.}
%\ig{this seems like an artificial restriction to me. Does it really come from the literature?} YES, I INCLUDED THIS.

Considering the $2$-neighbor bootstrap percolation model and the parameter $\tbp(G)$ (the analogous to the parameter $t_{\tau}(G)$ defined above for our problem), Przykucki~\cite{Przykucki} determined the value of the maximum percolation time on the hypercube~$2^{[n]}$ as a function of~$n$, and Benevides and Przykucki~\cite{benevides-EJC2013,benevides-SIDMA} obtained similar results for the square grid $[n]^2$.
It was also proved that deciding whether~$\tbp(G)\geq k$ is polynomial-time solvable for fixed $k\leq 3$~\cite{wg2014-tcs}, and \NP-complete for fixed~$k\geq 4$~\cite{benevides2015}.
In bipartite graphs, it is polynomial-time solvable for fixed $k\leq 4$ and \NP-complete for fixed $k\geq 5$~\cite{wg2014-tcs}.
Finally, it was proved in~\cite{MarcilonS18-tw} that $2$-\textsc{Neighbor Bootstrap Percolation-time} is \W[1]-hard parameterized  by the treewidth of the input graph. Clearly, all these hardness results extend to \textsc{Generalized TSS-time}. However, they cannot be applied directly to \textsc{TSS-time}, since all these hardness reductions use vertices of degree 1, which have an important role in them, and which are not allowed in our setting when all thresholds are equal to 2.
%\ig{this paragraph is too detailed, we don't need to state explicitly all these results}

\magenta{In the \pname{Geodesic (resp. Monophonic) Convexity-time} problem, threshold functions are not taken into account, and at any step of the activation process, each inactive vertex gets activated if it is in a shortest (resp. induced) path between two activate vertices. The maximum activation time obtained in these processes
%, denoted by $t_{gd}$ and $t_{mo}$,
has been studied~\cite{hn1981,benevides-geo,eurinardo}. For each parameter, deciding if its value is greater or equal to $k$ (for fixed $k$) is \NP-complete when $k \ge 2$ and $k \ge 1$, respectively, and the input graph is bipartite~\cite{benevides-geo,eurinardo}. Both problems are polynomial-time solvable for distance-hereditary graphs~\cite{benevides-geo}.}

With respect to the parameterized complexity of these problems, the published articles mainly focus on the generalized version of the \textsc{TSS-size} problem, for which there are no degree restrictions on the threshold function, denoted by \pname{Generalized TSS-size}, and the cases in which the maximum threshold is small or the threshold of every vertex is at least half of its degree. Namely, \pname{Generalized TSS-size} is \b{\FPT parameterized by the size of a minimum vertex cover~\cite{nichterlein2013,abs-1812-01482}, and \W[1]-hard for each of the following parameterizations: distance to cluster~\cite{chopin2014}, neighborhood diversity~\cite{dvork2018}, and distance to forest and pathwidth~\cite{nichterlein2013}.}
The case when all thresholds are exactly half of the degree for each vertex is also \W[1]-hard parameterized by pathwidth~\cite{chopin2014}.
For constant thresholds, the problem becomes~\FPT parameterized by distance to cluster~\cite{chopin2014}, by neighborhood diversity~\cite{dvork2018}, and by treewidth~\cite{BENZWI2011}. Ben-Zwi et al.~\cite{BENZWI2011} also proved that \pname{TSS-size} cannot be solved in $\Ocal(n^{\sqrt{\tw(G)}})$ time, where $n$ and $\tw(G)$ denote the number of vertices and the treewidth of the input graph $G$, respectively.
Recently, Hartmann~\cite{Hartmann18} gave an \FPT algorithm for~\pname{TSS-size} parameterized by clique-width \b{and the maximum value of the threshold function}. Cicalese et al.~\cite{CICALESE20141,CICALESE201540} considered the problem in which the number of rounds of the process is bounded. For graphs of bounded clique-width, given parameters $a$, $b$, $\ell$, they gave polynomial-time algorithms to determine whether there exists a target set of size $b$, such that at least $a$ vertices are activated in at most $\ell$ time steps.

%\ig{
%IMPORTANT: about the \textsc{Target Set Selection} problem (that is, when the goal is to minimize the size of a target set), I think we should also cite \FPT results, in order to motivate our \FPT algorithms. In particular, we don't cite this nice article of Hartmann~\cite{Hartmann18}. In page 2 of arXiv version of this article (\url{https://arxiv.org/pdf/1710.00635.pdf}), a lot of \FPT results are mentioned.\vini{Done!}\\
%By the way, in it is NOT required that $\tau(v) \leq d(v)$.\vini{I have the same doubt. This is the first time I hear about this restriction.} Did we impose this constraint here, or was it also imposed in other articles before?\vini{By the way, now there are two versions: the normal and the generalized one.} Please mention that.
%We should also discuss the results in~\cite{BENZWI2011} about the parameterized complexity of parameterized by treewidth}

\smallskip

% ---------------------------------------------------------------
% ---------------------------------------------------------------
% ---------------------------------------------------------------
% ---------------------------------------------------------------
% ---------------------------------------------------------------
\textbf{Our results and techniques}. In this paper we initiate an analysis of the computational complexity of the \pname{TSS-time} problem, in particular from the viewpoint of parameterized complexity.
%\rud{Let us call 2-\textsc{Neighbor Bootstrap Percolation time} the problem of, given a graph $G$ and an integer $k$, deciding if $\tbp(G)\geq k$.}{}
We start by showing that reductions of \cite{benevides2015} and \cite{MarcilonS18-tw} for the 2-\textsc{Neighbor Bootstrap Percolation-time} problem can be easily adapted in order to prove that \pname{TSS-time} is \NP-hard in planar graphs \b{and graphs of bounded degree}, and \W[1]-hard when parameterized by the treewidth of the input graph (\cref{corol-W1}).
%\ig{we haven't defined this problem, right?}
We then provide \NP-completeness results for fixed values of the activation time~$k$ and the value $\tau^*= \max_{v \in V(G)}\tau(v)$. Namely, by adapting another reduction in~\cite{benevides2015} from \textsc{$3$-Sat}, we prove (\cref{thm:NPC}) that \textsc{TSS-time} is \NP-complete in general graphs for any fixed $k \geq 4$ and $\tau^*=2$. This result is sharp in terms of $\tau^*$ since, as we observe in \cref{lem:tau1}, the problem can be easily solved in polynomial time when $\tau^*=1$. By reducing from the \textsc{Restricted Planar $3$-Sat} problem~\cite{DahlhausJPSY94} instead of \textsc{$3$-Sat} and modifying appropriately the planar embedding given by the incidence graph of the formula, we prove (\cref{thm:hard-apex}) that \textsc{TSS-time} remains \NP-complete for any fixed $k \geq 4$ and $\tau^*=2$ even if the input graph is an {\sl apex} graph, that is, a graph in which there exists a vertex whose removal yields a planar graph. Finally, by modifying the reduction of \cref{thm:NPC} by using bipartite gadgets, we prove (\cref{thm:NPC2}) that \textsc{TSS-time}
is \NP-complete in bipartite graphs for any fixed $k \geq 5$ and $\tau^*=2$.

\smallskip

Motivated by these \NP-completeness results, we study the parameterized complexity of the \pname{TSS-time} problem considering $k$ and $\tau^*$ as parameters.
%\ig{please, say either ``\textsc{TSS-time}'' or ``the \textsc{TSS-time} problem'', but not ``\textsc{TSS-time} problem'' (this occurs many times)}
We manage to provide a dichotomy on the complexity of \pname{TSS-time} when the input graph~$G$ belongs to a minor-closed graph class ${\cal C}$. Namely, we prove (\cref{thm:dichotomy}) that if ${\cal C}$ has {\sl bounded local treewidth} (cf.~\cref{sec:prelim} for the definition), then \textsc{TSS-time} is \FPT with parameters $k$ and $\tau^*$; otherwise it is \NP-complete for any fixed $k \geq 4$ and $\tau^*=2$. \b{Note that, as discussed above, \pname{TSS-time} is \NP-hard in planar graphs (even with $\tau^*=2$), which is a minor-closed graph class of bounded local treewidth, and therefore in our complexity dichotomy,  ``\FPT''  cannot be replaced by ``solvable in polynomial time''}. It is pertinent to mention here that the title of the article of Ben-Zwi et al.~\cite{BENZWI2011} is ``{\sl treewidth governs the complexity of target set selection}'', referring to the \textsc{TSS-size} problem. In this spirit, one of the the main conclusions of our article is that ``{\sl local treewidth governs the complexity of target set selection with maximum activation time}''. Let us now discuss how we prove \cref{thm:dichotomy}.

In order to prove this dichotomy, on the positive side we provide  (\cref{thm:FPT}) an \FPT algorithm \b{for the generalized version of the problem, namely \pname{GTSS-time}}, with parameters $k$ and $\tau^*$ when the input graph $G$ belongs to a graph class ${\cal C}$ of bounded local treewidth (not necessarily minor-closed). In order to do this, we first observe (\cref{lem:equal-at-least}) that, for any instance $(G,\tau,k)$ of \pname{GTSS-time},
%(in fact, slightly stronger, for \pname{GTSS-time}),
$t_{\tau}(G) \geq k$ if and only if there exists a target set activating $G$ at time {\sl exactly}~$k$. We then show (\cref{lem:balls}) that \pname{GTSS-time} on an $n$-vertex graph $G$ can be reduced to solving $n$ instances having as input the graph induced in $G$ by the $k$-th neighborhood of each vertex of $G$. The crucial observation is that, when $G$ belongs to a class of bounded local treewidth, these auxiliary graphs have treewidth bounded by a function of $k$. With this at hand, we show (\cref{lem:MSOL}) that \pname{GTSS-time} can be expressed by a monadic second-order logic formula whose length depends only on~$k$ and $\tau$, and applying Courcelle's Theorem~\cite{Courcelle90} on the linearly many bounded-treewidth auxiliary graphs yields the desired~\FPT algorithm. Note that, since we deal with the generalized version of the \pname{TSS-time} problem, our \FPT algorithm also applies to the 2-\textsc{Neighbor Bootstrap Percolation} problem.
As particular cases of graph classes with bounded local treewidth, the existence of an \FPT algorithm for 2-\textsc{Neighbor Bootstrap Percolation} with parameter $k$ in graphs with bounded maximum degree was already known \cite{MarcilonS18}, but no \FPT algorithm in planar graphs (or, more generally, graphs of bounded genus) existed prior to our work.
Note that 2-\textsc{Neighbor Bootstrap Percolation} has been proved to be \NP-complete in planar graphs by Benevides et al.~\cite{benevides2015}. In this \NP-completeness proof, the authors say that ``our proof does not work when the time is fixed''; the \FPT algorithm of \cref{thm:FPT} provides a solid explanation for that.

As for the hardness part of our complexity dichotomy, we critically use a result of Eppstein~\cite{Eppstein00} stating that, for minor-closed graph classes, having bounded local treewidth is equivalent to excluding some apex graph. Now, if ${\cal C}$ is a minor-closed graph class of unbounded local treewidth, the previous result implies that ${\cal C}$ contains all apex graphs, in particular those originated from our hardness result of \cref{thm:hard-apex} for apex input graphs, and therefore the \textsc{TSS-time} problem is
 \NP-complete in ${\cal C}$ for any fixed $k \geq 4$ and $\tau^*=2$. \b{Again, the same argument applies to 2-\textsc{Neighbor Bootstrap Percolation-time} (\cref{cor:P3-hard-apex}), hence the dichotomy in minor-closed graph classes holds for this problem as well.}

 %\ig{IMPORTANT: we have to say, maybe also in the abstract, that the above dichotomy also applies to the generalized version, and to $2$-\textsc{Bootstrap Percolation}. This is because both the \NP-hardness for apex graphs (\cref{thm:hard-apex}) and the \FPT algorithm (\cref{thm:FPT}) apply to the generalized version and to
 %$2$-\textsc{Bootstrap Percolation} as well. If you want, I can write this fact at the end, once I will take the token back.}
%\vini{I think that in the abstract we have to define many things before to say that it is also valid to these cases, such as defining the generalized version and cite the \pname{$2$-Bootstrap Percolation} problem. Maybe it is better to emphasize in this subsection.}

\smallskip

Finally, we provide (\cref{teo-trees1}) an $\Ocal(n)$-time algorithm  for the maximization version of \pname{TSS-time} in trees, that is, for finding the maximum activation time of \b{a target selection model $(T,\tau)$ where $T$ is a tree}. %ree $T$ for any threshold function~$\tau$ in $T$.
Note that the \FPT algorithm of \cref{thm:FPT} implies that deciding whether $t_{\tau}(T) \geq k$ for a tree $T$ (which has treewidth one) can be solved in time $f(k, \tau^*)\cdot n^{\Ocal(1)}$ for some function~$f$, but we provide a stronger result by showing that also the {\sl maximum} activation time of a tree can be computed in polynomial (even linear) time. In order to achieve this, we prove (\cref{lema1b} and \cref{teo-trees1}) that every path such that any internal vertex $v$ satisfies $\tau(v)<d(v)$ (we say that such a $v$ is \emph{non-saturated}) can be activated one vertex per time step by some target set. With this, we prove that the maximum activation time in a tree is equal to the size of a maximum path such that all internal vertices are non-saturated. One interesting point here is that the threshold values are not important, but only whether a vertex is saturated or not.
%\ig{IMPORTANT: please sketch the main ideas of the algorithm for trees here, by also saying whether we got inspired by some previous proof or not}
We generalize this algorithm (\cref{teo-trees2}) to the maximization version of \pname{GTSS-time}, namely, we provide an $\Ocal(n^2)$-time algorithm for finding the maximum activation time of a tree $T$ and a generalized threshold function $\tau$ in $T$.
%\ig{IMPORTANT: please cite a precise theorem here -- it is missing in \cref{sec:PolyTreesG}}
%\ig{shall we say here that we also allow $\tau(v)=0$ in \cref{sec:PolyTrees-G}? This is supposed to be also allowed in the original \textsc{TSS-time} problem}.
The main idea is that any target set must contain the set $V_{\sf f}$ of ``forced'' vertices containing any vertex $v$ with $\tau(v)>d(v)$. With this, we start the activation process from $V_{\sf f}$, obtaining the set $H(V_{\sf f})$ containing the vertices that can be activated by $V_{\sf f}$. We then look for certain paths representing an activation sequence, whose first vertices are activated by $V_{\sf f}$ and whose last vertices are non-saturated. In this case, the threshold values are important, since $H(V_{\sf f})$ depends on them.
%\ig{IMPORTANT: please sketch the main ideas of the algorithm for trees here, by also saying whether we got inspired by some previous proof or not}

%\ig{Also, what about $\tau^* = 1$? It corresponds to finding a vertex maximizing the radius (in each connected component of the input graph $G$),  right? Thus, the problem can be solved in polynomial time. We should mention that our \NP-completeness results when $\tau^* = 2$ are tight with respect to the value of $\tau^*$.}

\medskip

\noindent\textbf{Organization}. In \cref{sec:prelim} we provide basic preliminaries about graphs, convexity, parameterized complexity, graph minors, (bounded local) treewidth, and monadic second-order logic. In \cref{sec:NPcomp} we present our \NP-completeness results, and in \cref{sec:FPT-Alg} we provide the \FPT algorithms for graphs of bounded local treewidth. Altogether, the results in \cref{sec:NPcomp} and \cref{sec:FPT-Alg} yield the complexity dichotomy for graph classes of bounded local treewidth. \cref{sec:PolyTrees} is devoted to the polynomial-time algorithms for trees. %, and in \cref{sec:PolyTreesG} we generalize this result by allowing thresholds equal to 0 and  pairs $(G,\tau)$ where $\tau(v)>d(v)$ is possible.
We conclude the paper in \cref{sec:concl} with some directions for further research.

% ---------------------------------------------------------------
% ---------------------------------------------------------------
% ---------------------------------------------------------------
% ---------------------------------------------------------------
% ---------------------------------------------------------------
%\subsection{Related work and some notation}
%\label{sec:related-notation}

%\ig{IMPORTANT: the whole discussion of this \cref{sec:related-notation} is very nice, but do we really need to speak about so many different convexities in our paper? We are not doing a survey. What we need is to motivate more the variant of the problem that we are studying. Why should one be interested in maximizing the activation time? Which applications does it have? We should really insist on that}
%I THINK YOU ARE RIGHT. THE MOST IMPORTANT POINT HERE IN MY OPINION IS TO RELATE \pname{TSS} WITH GENERAL CONVEXITIES. MAYBE THIS WILL BE IMPORTANT FOR FUTURE PAPERS. I WILL REMOVE EVERYTHING NOT DIRECTED RELATED TO TSS.

\section{Preliminaries}
\label{sec:prelim}

%\ig{This section is new -- I added the definitions we use in the paper}

\noindent \textbf{Graphs}. We refer the reader to~\cite{Diestel12} for basic background on graph theory, and recall here only some useful definitions. We consider only undirected graphs without loops nor multiple edges. We will use $n$ and $m$ for denoting the number of vertices and edges, respectively, of the input graph of the problem under consideration.
We denote by $uv$ an edge between vertices $u$ and $v$. For a graph $G$ and a vertex set~$S \subseteq V(G)$, we use the shortcut $G \setminus S$ to denote $G[V(G) \setminus S]$.
The \emph{distance} between two vertices $u$ and $v$ in a graph $G$ is the number of edges of a shortest path between $u$ and $v$. The \emph{diameter} of  $G$ is the maximum distance over all pairs of vertices of $G$. For a vertex $v$ in~$G$ and a positive integer $k$, we denote by $N_k(v)$ the set of vertices of $G$ within distance at most~$k$ from $v$, and we let $N_k[v]=N_k(v) \cup \{v\}$. We abbreviate $N_1(v)$ and $N_1[v]$ as $N(v)$ and~$N[v]$, respectively, and we let $d(v) = |N(v)|$ be the \emph{degree} of $v$ in $G$. A \emph{tree} is a connected acyclic graph, and a \emph{leaf} in a tree is a vertex of degree one.

For two non-negative integers $a$ and $b$, we denote by $[a,b]$ the set containing every integer~$c$ such that $a \leq c \leq b$ and we let $[a] = [1,a]$. If a set $S$ is partitioned into pairwise disjoint sets $S_1, \dotsc, S_k$, we denote it by $S = S_1 \uplus \dots \uplus S_k$.

%\ig{We shall mention the FPT results (and references) given here: https://arxiv.org/pdf/1812.01482.pdf}

\medskip
\noindent \textbf{Convexity.} Activation problems appear in the literature under a number of different names, such as $r$-neighbor bootstrap percolation~\cite{riedl,morris,balogh4,balogh3,neural2,balogh2,balogh,chalupa,holroyd}, dynamic monopolies~\cite{BESSY2019,soltani2019,ZAKER2012,KHOSHKHAH2014,CHANG2013}, irreversible conversion \cite{CDPRS2011,dmtcs:3952,cs1,TAKAOKA2015}, or graph convexities, and were studied by researchers of various fields.
As mentioned in the introduction, in the particular case in which all thresholds are equal to $2$, generalized \pname{TSS} model is also called $2$-neighbor bootstrap percolation or \emph{$P_3$-convexity}, which is studied in the field of graph convexities.

A finite {\it graph convexity}~\cite{ve} is a pair $(G,\mathcal{C})$ consisting of a finite simple graph~$G=(V,E)$ and a set $\mathcal{C}$ of subsets of $V$ (called \emph{convex sets}) satisfying that~$\emptyset, V\in\mathcal{C}$ and that if $C_1,C_2\in \mathcal{C}$, then $C_1\cap C_2\in \mathcal{C}$. In words, $\emptyset$ and $V$ are convex sets and the intersection of convex sets is a convex set. The {\em convex hull}~$H_\mathcal{C}(S)$ of a set $S$ is the minimum convex set containing $S$, that is,~$H_\mathcal{C}(S)$ is the intersection of all convex sets containing $S$. When $H_\mathcal{C}(S) = V$ then $S$ is a {\it hull set} of $G$.
Some well-studied graph convexities are the so-called \emph{path convexities}, such as the $P_3$-convexity~\cite{er1972}, geodesic convexity~\cite{faja}, and monophonic convexity~\cite{du1988}.

In the following, we show that instances of the generalized \pname{TSS} model
induce graph convexities in most cases. Let $(G,\tau)$ be an instance of the generalized \pname{TSS} model, where~$G=(V,E)$ is a graph and $\tau:V\to\NN$ is a generalized threshold function. For every set $S\subseteq V$, let the \emph{interval} $I_\tau(S)\supseteq S$ be the union of the set $S$ with the set of all vertices $v$ outside $S$ which have $\tau(v)$ neighbors in $S$. From this, let $\mathcal{C_\tau}$ be the family of subsets $S$ of $V$ such that~$I_\tau(S)=S$ (that is, no vertex $v$ outside~$S$ has $\tau(v)$ neighbors in $S$).

\begin{lemma}\label{lem-convexity}
Given an instance $(G,\tau)$ of the generalized \textsc{TSS} model, where $G=(V,E)$ is a graph and~$\tau:V\to\NN$ is a generalized threshold function in $G$, the pair $(G,\mathcal{C}_\tau)$ is a graph convexity if and only if $V=\emptyset$ or all thresholds are strictly positive.
%\ig{this wasn't known before?} I NEVER SAW CONVEXITY IN SUCH A GENERAL SETTING.
\end{lemma}
\begin{proof}
We have to prove that $\emptyset, V\in\mathcal{C}_\tau$ and that if $C_1,C_2\in \mathcal{C}_\tau$, then $C_1\cap C_2\in \mathcal{C}_\tau$. Clearly~$V\in\mathcal{C}_\tau$ by vacuity, since there is no vertex outside $V$.
Thus, if $V=\emptyset$, $(G,\mathcal{C}_\tau)$ is a graph convexity, since the only subset~$S$ of $V$ is $S=\emptyset=V$. So assume that $V\ne\emptyset$.

First consider that there is a vertex $v$ with threshold $\tau(v)=0$.
With this, we have that~$S=\emptyset$ is not convex, since $v\not\in S=\emptyset$ and $v$ has $\tau(v)=0$ neighbors in $S=\emptyset$. Then~$(G,\mathcal{C}_\tau)$ is not a graph convexity.

Now assume that all thresholds are strictly positive.
Therefore $S=\emptyset$ is convex, since any vertex $v$ does not have $\tau(v)$ neighbors in $S=\emptyset$. Finally, consider two sets $S_1,S_2\in\mathcal{C}_\tau$ and let $S=S_1\cap S_2$. If $S_1=V$ or $S_2=V$, then $S=S_2$ or $S=S_1$, respectively, and hence~$S\in\mathcal{C}_\tau$. So assume that $S_1\ne V$ and $S_2\ne V$, and let~$v\in V\setminus S$. Then $v\not\in S_1$ or~$v\not\in S_2$. Consider~$v\not\in S_1$. Since $S_1\in\mathcal{C}_\tau$, $v$ does not have $\tau(v)$ neighbors in~$S_1$ and consequently does not have $\tau(v)$ neighbors in $S$. The case $v\not\in S_2$ is analogous. Then~$S=S_1\cap S_2\subseteq\mathcal{C}_\tau$.
\end{proof}

In this context, we can also define the \emph{activation time} $t_{\tau}(S)$ of a vertex subset~$S$ (not necessarily a target set) as the minimum $t$ such that $I_\tau^{t+1}(S)=I_\tau^t(S)$, where $I_\tau^k(S)$ is the $k$-th iterate of the interval function, defined recursively as $I_\tau^0(S)=S$ and $I_\tau^k(S)=I_\tau^{k-1}(I_\tau(S))$ for~$k\geq 1$. We can also define $H_\tau(S)$ as $H_\tau(S)=I_\tau^{t_{\tau}(S)}(S)$. This definition of $H_\tau(S)$ is useful even when $(G,\mathcal{C}_\tau)$ is not a graph convexity (for example, when some thresholds are 0). Recall that a vertex subset $S$ is a target set if~$H_\tau(S)=V$. In this paper, we will use these notations $I_\tau(S)$ and $H_\tau(S)$. When $\tau$ is clear in the context, the subscript will be removed from the notations $I(S)$ and $H(S)$.

%Convexity spaces form a classical topic, studied in some different branches of mathematics. The study of convexities applied to graphs has started later, about 50 years ago. Then the convexity parameters motivated the definition of some graph parameters, whose study has been one of the central issues in graph convexities.
The study of complexity aspects related to the computation of graph convexity parameters have been the main goal of various recent papers~\cite{DOURADO2017,benevides-geo,barbosa-2012,eurinardo,MarcilonS18,BUENO201822}.
From \cref{lem-convexity}, all known convexity parameters, such as the Carath\'eodory number, the Radon number, the Helly number and the convexity number~\cite{DOURADO2017} are meaningful in the \pname{TSS} model and can be investigated in this context.
%\ig{please move this discussion to the conclusions section}

\medskip
\noindent \textbf{Parameterized complexity}. We refer the reader to~\cite{DF13,CyganFKLMPPS15} for basic background on parameterized complexity, and we recall here only the definitions used in this article. A \emph{parameterized problem} is a language $L \subseteq \Sigma^* \times \mathbb{N}$. For an instance $I=(x,k) \in \Sigma^* \times \mathbb{N}$, $k$ is called the \emph{parameter}.

A parameterized problem $L$ is \emph{fixed-parameter tractable} (\FPT) if there exists an algorithm~$\mathcal{A}$, a computable function $f$, and a constant $c$ such that given an instance $I=(x,k)$,~$\mathcal{A}$ (called an \FPT \emph{algorithm}) correctly decides whether $I \in L$ in time bounded by $f(k) \cdot |I|^c$.
%For instance, the \textsc{Vertex Cover} problem parameterized by the size of the solution is {\sf FPT}.
	
%A parameterized problem $L$ is in {\sf XP} if there exists an algorithm $\mathcal{A}$ and two computable functions $f$ and $g$ such that given an instance $I=(x,k)$, $\mathcal{A}$  (called an {\sf XP} \emph{algorithm}) correctly decides whether $I \in L$ in time bounded by $f(k) \cdot |I|^{g(k)}$. For instance,  the \textsc{Clique} problem parameterized by the size of the solution is in  {\sf XP}.

Within parameterized problems, the class \W[1] may be seen as the parameterized equivalent to the class \NP~of classical decision problems. Without entering into details (see~\cite{DF13,CyganFKLMPPS15} for the formal definitions), a parameterized problem being \W[1]-\emph{hard} can be seen as a strong evidence that this problem is {\sl not} \FPT.
%The canonical example of {\sf W}[1]-hard problem is \textsc{Clique}  parameterized by the size of the solution.

\medskip
\noindent \textbf{Minors, treewidth, and bounded local treewidth}. A graph $H$ is a \emph{minor} of a graph~$G$ if $H$ can be obtained from a subgraph of $G$ by contracting edges. A graph class ${\cal C}$ is \emph{minor-closed} if whenever a graph $G$ belongs to ${\cal C}$, all its minors belong to ${\cal C}$ as well.  A graph is \emph{planar} if it can be drawn in the plane so that its edges may intersect only in the extremities. A graph $G$ is an \emph{apex graph} if it contains a vertex whose removal from $G$ results in a planar graph.

Let $k \geq 0$ be an integer. A graph $G$ is a \emph{$k$-tree} if $G$ can be obtained from a clique of size~$k+1$ by repeatedly adding vertices adjacent to a clique of size $k$ of the current graph. The \emph{treewidth} of a graph $G$, denoted by $\tw(G)$, is the smallest integer $k$ such that $G$ is a subgraph of a $k$-tree.

A graph class ${\cal C}$ has \emph{bounded local treewidth} if there exists a function $f:\NN \to\NN$ such that, for every graph $G \in {\cal C}$, every vertex $v \in V(G)$, and every positive integer $k$, $\tw( G [N_k[v]]) \leq f(k)$. Examples of graph classes of bounded local treewidth are graphs of bounded treewidth, graphs of bounded degree, planar graphs, or graphs of bounded genus; see~\cite{Grohe03} for more on bounded local treewidth. The following theorem of Eppstein~\cite{Eppstein00} states that, for minor-closed graph classes, having bounded local treewidth is equivalent to excluding some apex graph.

\begin{theorem}[Eppstein~\cite{Eppstein00}]
\label{thm:Eppstein}
Let ${\cal C}$ be a minor-closed graph class.  Then ${\cal C}$ has bounded local treewidth if and only if ${\cal C}$ does not contain all apex graphs.
\end{theorem}

\medskip
\noindent \textbf{Monadic second-order logic of graphs}. The syntax of \emph{monadic second-order logic} (\MSO)  of graphs includes the logical connectives $\vee$, $\wedge$, $\neg$, variables for vertices, edges, sets of vertices and sets of edges, the quantifiers $\forall, \exists$ that can be applied to these variables, and the binary relations expressing whether a vertex or an edge belong to a set, whether an edge is incident to vertex, whether two vertices are adjacent, and whether two sets are equal. \MSOone is the restriction of \MSO where only quantification over sets of vertices (but not edges) is allowed.
The following result of Courcelle~\cite{Courcelle90}
%, as well as one of its several optimization variants~\cite{ArnborgLS91},
is one of the most widely used results in the area of parameterized complexity.
%\karol{Is it obvious how the size of a formula is measured? Wouldn't it be ok to say it here?}

\begin{theorem}[Courcelle~\cite{Courcelle90}]
\label{thm:Courcelle}
Checking whether an \MSO formula $\phi$ holds on an $n$-vertex graph of treewidth at most $\tw$ can be done in time $g(\phi, \tw) \cdot n$, for a computable function $g$.
%Moreover, within the same running time, one can find a vertex or edge set of $G$ of maximum or minimum size that satisfies $\varphi$.
\end{theorem}

\section{\NP-completeness results for the TSS-time problem}
\label{sec:NPcomp}

In this section, we prove \NP-completeness results  for the \pname{TSS-time} problem. Namely, we prove \NP-completeness for general graphs in \cref{thm:NPC}, for apex graphs in \cref{thm:hard-apex}, and for bipartite graphs in \cref{thm:NPC2}.
We begin by proving membership in \NP.

%\karol{Isn't better do give a different notation for the parameter defined in the next paragraph, to avoid any confusion with the one defined in the beginning  of the FPT section?}
%
%\magenta{\sout{We first prove that, for every set $S_0$, it is possible to compute the parameter $t_\tau(v,S_0)$ for all vertices $v$ in  time $\Ocal(m+n)$, where $t_\tau(v,S_0)$ is the activation time of vertex~$v$ in the process initiated by $S_0$ (set $t(v,S_0)=\infty$ if $S_0$ cannot activate $v$). Also, let $t_\tau(S_0)=\max_{v\in V(G)}\{t_\tau(v,S_0)\}$. Thus, a vertex set $S_0$ is a target set if \b{and only if} $t_\tau(S_0)<\infty$.}}

\magenta{Let $t_\tau(v,S_0)$ be the activation time of vertex~$v$ in the process initiated by $S_0$ (set $t(v,S_0)=\infty$ if $S_0$ cannot activate $v$). We first prove that, for every set $S_0$, it is possible to compute $t_\tau(v,S_0)$ for every vertex $v$ in time $\Ocal(m+n)$. Also, let $t_\tau(S_0)=\max_{v\in V(G)}\{t_\tau(v,S_0)\}$. Thus, a vertex set $S_0$ is a target set if \b{and only if} $t_\tau(S_0)<\infty$.}

\begin{lemma}\label{lem:t(v,S0)}
Let $G$ be a graph and $\tau$ be a generalized threshold function in $G$.
Given a set~$S_0\subseteq V(G)$, it is possible to compute $t_\tau(v,S_0)$ for all vertices $v$ of $G$ in  time~$\Ocal(m+n)$.
\end{lemma}

\begin{proof}
Consider the following algorithm. Let $Q$ be an empty queue and $t$ an array such that~$t[v]=\infty$ for any $v\in V(G)\setminus S_0$. For each $v\in S_0$, set $t[v]=0$ and enqueue $v$ in $Q$. For each vertex $v\not\in S_0$ with threshold 0, set $t[v]=1$ and enqueue $v$ in $Q$.
%\orange{Finally, let $A$ be a Boolean vector such that $A[v]$ indicates whether $v \in V(G)$ is activated or not. For each~$v\in Q$, set $A[v]$=true and set $A[v]$=false for each $v \notin Q$.}

\vspace{-.15cm}
\begin{itemize}
\item Algorithm \pname{Activation-Times} (set $S_0$)
\item[1]\ \textbf{while} $Q\ \ne\ \emptyset$ \textbf{do}
\item[2]\ \ \ \ \ \ $v\ \leftarrow$\ \pname{Dequeue}($Q$)
\item[3]\ \ \ \ \ \ \textbf{for} each neighbor $u$ of $v$ with $t[u]= \infty$ \textbf{do}
\item[5]\ \ \ \ \ \ \ \ \ \ \ $\tau(u)\ \leftarrow\ \tau(u)-1$
\item[6]\ \ \ \ \ \ \ \ \ \ \ \textbf{if} $\tau(u)=0$ \textbf{then}
\item[7]\ \ \ \ \ \ \ \ \ \ \ \ \ \ \ \ $t[u]\ \leftarrow \ t[v]+1$
\item[9]\ \ \ \ \ \ \ \ \ \ \ \ \ \ \ \ \pname{Enqueue}($Q,u$)
\item[10]\ \textbf{return}  array $t$
%\orange{\item[4]\ \ \ \ \ \ \ \ \ \ \ \textbf{if} $A[u]=$\ false \textbf{then}
%\item[5]\ \ \ \ \ \ \ \ \ \ \ \ \ \ \ \ $\tau(u)\ \leftarrow\ \tau(u)-1$
%\item[6]\ \ \ \ \ \ \ \ \ \ \ \ \ \ \ \ \textbf{if} $\tau(u)=0$ \textbf{then}
%\item[7]\ \ \ \ \ \ \ \ \ \ \ \ \ \ \ \ \ \ \ \ \ $t[u]\ \leftarrow \ t[v]+1$
%\item[8]\ \ \ \ \ \ \ \ \ \ \ \ \ \ \ \ \ \ \ \ \ $A[u]\ \leftarrow$ \ true
%\item[9]\ \ \ \ \ \ \ \ \ \ \ \ \ \ \ \ \ \ \ \ \ \pname{Enqueue}($Q,u$)
%\item[10]\ \textbf{return} vector $t$}
\end{itemize}

%\ig{Vini, you already made the changes here, right?}
%\vini{Yes, it is in orange.}

%\magenta{\sout{
%Since the Algorithm \pname{Activation-Times($S_0$)} simulates the activating process and every edge is analyzed at most twice, we have time $\Ocal(m+n)$.}}

\magenta{The above algorithm simulates the activation process. Since every edge is analyzed at most twice, \pname{Activation-Times($S_0$)} runs in time $\Ocal(m+n)$.}
\end{proof}

With this, we have membership in \NP\ for \pname{GTSS-time}.
\begin{corollary}\label{lem:facilNP}
The \pname{GTSS-time} problem is in \NP.
\end{corollary}

\begin{proof}
Given an instance $(G,\tau,k)$ of \pname{GTSS-time}, a certificate (or proof) for it is a vertex subset $S_0\subseteq V(G)$. With the Algorithm \pname{Activation-Times}, $t_\tau(v,S_0)$ can be computed for all $v\in V(G)$ in $\Ocal(m+n)$ time. With this, $t_\tau(S_0)$ can be computed and compared with $k$.
\end{proof}

Before moving to the hardness results, consider first the \pname{GTSS-time} problem with all thresholds being at most 1. Given a connected graph $G$, it is easy to see that one vertex is sufficient to activate all vertices. If there are vertices with threshold 0, then $\emptyset$ is a target set and then $t_\tau(G)=t_\tau(\emptyset)$, which can be computed in linear time by the algorithm \pname{Activation-Times} with $S_0=\emptyset$. Otherwise, every single vertex is a target set and then $t_\tau(G)$ is the diameter of $G$, which can be computed in time $\Ocal(m\cdot n)$. With this, we have the following:

%\karol{I think $\tau^*$ should be defined in other place than the abstract or "our techniques". First paragraph of Introduction? Preliminaries?}

\begin{lemma}\label{lem:tau1}
Let $G$ be a graph and $\tau$ be a generalized threshold function in $G$ satisfying $\tau^*\leq 1$.
Then~$t_{\tau}(G)$ can be computed in time $\Ocal(m\cdot n)$.
Thus, \pname{GTSS-time} is $\Ocal(m\cdot n)$-time solvable if $\tau^*\leq 1$.
%\orange{\sout{, where~$\tau^*=\max_{v\in V(G)}\{\tau(v)\}$}} \vini{$\tau^*$ is already defined.}
\end{lemma}

Now let us consider the case where all thresholds are equal to 2.
As mentioned in the introduction, the $2$-neighbor bootstrap percolation model is equivalent to the generalized \pname{TSS} model with all thresholds equal to $2$.
With this, let us list again shortly the existing hardness results for the $2$-\textsc{Neighbor Bootstrap Percolation-time} problem:
%\karol{NP-hardness is the best term to be used here? It seems a bit weired when reading the complete sentence.}
\NP-hardness in planar graphs \cite{benevides2015}, \NP-hardness in general graphs for fixed $k=4$ \cite{benevides2015}, \NP-hardness in bipartite graphs for fixed $k=5$ \cite{wg2014-tcs}, \NP-hardness in bounded degree graphs for $k=\Theta(\log n)$ \cite{MarcilonS18}, and \W[1]-hardness when parameterized  by  treewidth \cite{MarcilonS18-tw}.
All these hardness results also apply to \textsc{GTSS-time} with all thresholds equal to $2$, but cannot be extended directly to \pname{TSS-time}, since they use many vertices of degree~1, which have activation time 0 and are important to control the maximum activation time.

However, except in the case of bipartite graphs, it is possible to apply local changes to all these reductions by replacing every vertex $p$ of degree 1 by two adjacent vertices $p_1$ and~$p_2$, forming a triangle with the original neighbor $q$ of $p$ (this is the reason why this replacement does not work in bipartite graphs).
Consider any of the reductions mentioned above, and let $(G,k)$ be the original constructed instance of $2$-\textsc{Neighbor Bootstrap Percolation time}, where~$G$ is the graph and $k$ is the desired activation time. Let $(G',k')$ be the instance where $G'$ is the graph obtained from $G$ with this modification and $k'=k+1$.
Since all thresholds are~2, we have that, for each vertex $p$ of degree~1 in $G$, at least one of $p_1$ or $p_2$ (say $p_1$ w.l.g.) must be in any target set of $G'$ and will play in $G'$ the same role as $p$ in $G$.
Thus any target set $S$ of $G$ induces a target set $S'$ of $G'$, which activates $p_2$ in one time step more than $q$.
Now consider a target set $S'$ of $G'$. One important point in all these reductions is that the neighbor $q$ of any vertex $p$ of degree~1 is always activated (with the help of $p$) by a forced set of vertices which must belong to any target set (hull set in their terminology)\magenta{\ and this set remains forced when the above modification is applied in the construction}. That is, for any vertex $p$ of degree 1 in $G$, $S''=(S'\cup\{p_1\})\setminus\{p_2\}$ is also a target set of $G'$ with activation time greater or equal to the activation time of $S'$.
With this, we may assume that, for every vertex $p$ of degree~1 in $G$, $S'$
contains $p_1$ and does not contain $p_2$, and consequently it induces a target set $S$ in $G$ (just replacing $p_1$ by $p$ for any $p$ of degree~1 in $G$).
Finally, in all these reductions, all vertices are activated at time at most $k-1$ for any target set, except a special vertex $z$ which can be activated at time $k$ if the reduction is from a {\sf yes}-instance. Moreover, $z$ has exactly one neighbor $p$ of degree~1 in $G$ (in all these reductions) and consequently the corresponding neighbor $p_2$ of $z$ in $G'$ can be activated at time $k'=k+1$.
%\ig{perfect, now I totally agree with this explanation}
%
These modifications can be safely applied to the \NP-hardness reductions for bounded degree graphs with $k=\Theta(\log n)$ \cite{MarcilonS18} and planar graphs \cite{benevides2015}, and in the \W[1]-hardness reduction when parameterized by treewidth \cite{MarcilonS18-tw}, yielding the following corollary.

%\ig{maybe we can say that the following result is rather a corollary of the existing results?}\rud{DONE.}{}

\begin{corollary}\label{corol-W1}
The \pname{TSS-time} problem is \NP-hard in planar graphs, \NP-hard in graphs with maximum degree $\Delta$ for any fixed $\Delta\geq 4$ and $k=\Theta(\log n)$, and $\W[1]$-hard when parameterized by the treewidth of the input graph, even if all thresholds are equal to 2.
\end{corollary}

In the case of the reduction for fixed $k=4$ in general graphs~\cite{benevides2015}, which is from the \textsc{$3$-Sat} problem, this global argument does not work, since there is a unique vertex $z$ whose activation time is 3 or 4, depending on whether the \textsc{$3$-Sat} formula is satisfiable or not. However, by replacing every vertex $p$ of degree~1 by $p_1$ and $p_2$ as before, the activation time of $p_2$ is one more than the time of $q$ (the neighbor of $p$) and then the reduction fails for fixed $k=4$ (but works for $k=5$).
With a small additional change, the reduction can be corrected for $k=4$.
In the following, we present this modified reduction, where we also have to show that no additional vertex ($p_1$ or $p_2$) can be activated at time 4.
Moreover, although this reduction is similar to the one of \cite{benevides2015}, we present it in detail since we need to modify it in the proof of \NP-hardness for apex graphs shown in \cref{thm:hard-apex}.
%\ig{great explanation, I agree with it}

\begin{theorem}\label{thm:NPC}
The \textsc{TSS-time} problem is \NP-complete even restricted to instances $(G,\tau,k)$ such that $\tau(v) = 2$ for every $v \in V(G)$, and $k \geq 4$ is fixed.
\end{theorem}

\begin{proof}
We present a reduction from the \textsc{$3$-Sat} problem.
%\ig{we have to say that we assume that each clause contains \emph{exactly} 3 literals}.
Let $\varphi = ({\cal X},{\cal C})$ be an instance of \pname{$3$-SAT}, where
${\cal X}=\{x_1, \dotsc,x_n\}$ is the set of variables and $\mathcal{C}=\{C_1,\dotsc,C_m\}$ is the set of clauses.  We may assume that each clause contains exactly 3 literals.
For $i \in [n]$, we denote the three literals of clause~$C_i$ by $\ell_{i,1}$, $\ell_{i,2}$, and $\ell_{i,3}$. The constructed graph $G$ is described below.

%\bigskip
%\noindent {\bf Construction 1:}
For every clause $C_i$, add the gadget depicted in \cref{gadget1a}.
Let $U$, $W$, and $B$ be the sets containing all vertices $u_{i,p}$, $w_{i,p}$, and $b_{i,p}$ for $p \in [3]$, respectively.
Let $U_i=\{u_{i,1}, u_{i,2}, u_{i,3}\}$.
For every pair of complementary literals $\ell_{i,p}, \ell_{j,q}$ for $i,j \in [n]$ and $p,q \in [3]$,
%\magenta{\sout{(that is, $\ell_{i,p}, \bar{\ell}_{j,q}$)}}\karol{I don't think this notation is correct or necessary},
add a vertex $y_{(i,p),(j,q)}$ adjacent to $w_{i,p}$ and $w_{j,q}$.
Let $Y$ be the set of all vertices $y_{(i,p),(j,q)}$.
Finally, add six vertices $z,z_0,z_1,z_2,z_3,z_4$ and the edges $zz_0,z_0z_1,z_0z_2,z_0z_3,z_0z_4,z_1z_2$, and~$z_3z_4$. Also join $z$ with an edge to every vertex of $Y$ (see \cref{gadget1b}).
This completes the construction of the instance $(G,\tau,k)$ of \textsc{TSS-time}, where $k=4$ and $\tau(v) = 2$ for every $v \in V(G)$. Notice that $G$ does not contain vertices of degree~1, as required.
%\bigskip

%%%%%%%%%%%%%%%%%%%%%%%%%%%%%%%%%%%%
%%%%%%%%%%%%%%%%%%%%%%%%%%%%%%%%%%%%
\begin{figure}[h]
\centering\scalebox{0.75}{
\begin{tikzpicture}[scale=0.9]
%\tikzstyle{every node}=[font=\scriptsize]
\tikzstyle{vertex}=[draw,circle,fill=black!25,minimum size=15pt,inner sep=1pt]

\foreach \i [evaluate=\i as \angle using \i*120+90] in {1,2,3}
    \node[vertex] (u\i) at (\angle:1cm) {$u_{i,\i}$};
\foreach \i [evaluate=\i as \angle using \i*120+90] in {1,2,3}
    \node[vertex] (w\i) at (\angle:3cm) {$w_{i,\i}$};

\node[vertex,shift={(160:3cm)}] (a1) at (w1) {$a_{i,1}$};
\node[vertex,shift={( 20:3cm)}] (a2) at (w2) {$a_{i,2}$};
\node[vertex,shift={(180:3cm)}] (a3) at (w3) {$a_{i,3}$};
\node[vertex,shift={(180:3cm)}] (b1) at (w1) {$b_{i,1}$};
\node[vertex,shift={( 0 :3cm)}] (b2) at (w2) {$b_{i,2}$};
\node[vertex,shift={(200:3cm)}] (b3) at (w3) {$b_{i,3}$};

\foreach \x/\y in {1/2,2/3,3/1} \draw[-] (u\x) to (u\y);
\foreach \x/\y in {1/2,2/3,3/1}	\draw [-] (w\x) to [bend right=30] (w\y);
\foreach \x in {1,2,3} \draw[-] (w\x) -- (u\x);
\foreach \x in {1,2,3} \draw[-] (a\x) -- (b\x);
\foreach \x in {1,2,3} \draw[-] (w\x) -- (a\x);
\foreach \x in {1,2,3} \draw[-] (w\x) -- (b\x);

\node at (-1,-1.1) {\textcolor{blue}{0}};
\node at ( 1,-1.1) {\textcolor{blue}{3}};
\node at (.6, 1) {\textcolor{blue}{3}};
\node at (-2.2,-2.1) {\textcolor{blue}{1}};
\node at (2.2,-2.1) {\textcolor{blue}{2}};
\node at (0.7,3) {\textcolor{blue}{2}};

\node at (-4,3) {\textcolor{blue}{0}};
\node at (-3.8,2) {\textcolor{blue}{3}};
\node at (-6.7,-1.5) {\textcolor{blue}{2}};
\node at (-6.5,-0.3) {\textcolor{blue}{0}};
\node at (6.7,-1.5) {\textcolor{blue}{3}};
\node at (6.5,-0.3) {\textcolor{blue}{0}};

\end{tikzpicture}}
\caption{Gadget for a clause $C_i$. All thresholds are 2. Notice that every target set must contain at least one of $u_{i,1}, u_{i,2}, u_{i,3}$ and at least one of $a_{i,p}, b_{i,p}$ for $p \in [3]$. The blue numbers near the vertices show the times of an example of activation process. Vertices with time 0 belong to the target set.}\label{gadget1a}
\end{figure}
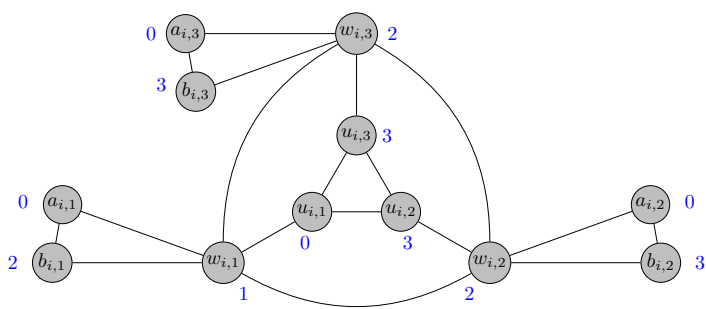
%%%%%%%%%%%%%%%%%%%%%%%%%%%%%%%%%%%%
%%%%%%%%%%%%%%%%%%%%%%%%%%%%%%%%%%%% 

Firstly notice that, for $i\in [m]$ and $p \in [3]$,  $a_{i,p}$ and $b_{i,p}$ cannot be activated only by $w_{i,p}$, since their degrees are equal to their thresholds (equal to two) and they are adjacent. That is, every target set must contain $a_{i,p}$ or $b_{i,p}$, say $a_{i,p}$ w.l.g.. The same argument applies to~$z_1$ and $z_2$ (say $z_1$ w.l.g.) and to $z_3$ and $z_4$ (say $z_3$ w.l.g.). From this, we have that $z_0$ is activated at time 1 and $z_2$ and $z_4$ are activated at time 2. The important fact here is that $z$ has a neighbor $z_0$ activated at time 1. Let $L$ be the set containing vertices $z_1$, $z_3$ and all vertices~$a_{i,p}$ for $i\in [m]$ and $p \in [3]$.

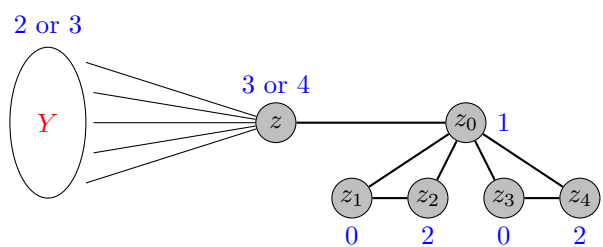
\begin{figure}[!bt]
\centering\scalebox{1.0}{
\begin{tikzpicture}[scale=1]
\tikzstyle{vertex}=[draw,circle,fill=black!25,minimum size=15pt,inner sep=1pt]

\node[vertex] (z)  at (0,1) {$z$};
\node[vertex] (z0) at (2.5,1) {$z_0$};
\node[vertex] (z1) at (1,0) {$z_1$};
\node[vertex] (z2) at (2,0) {$z_2$};
\node[vertex] (z3) at (3,0) {$z_3$};
\node[vertex] (z4) at (4,0) {$z_4$};

\path[-,thick] (z0) edge (z) edge (z1) edge (z2) edge (z3) edge (z4);
\path[-,thick] (z1) edge (z2) (z3) edge (z4);

\node at (1,-0.5) {\textcolor{blue}{0}};
\node at (2,-0.5) {\textcolor{blue}{2}};
\node at (3,-0.5) {\textcolor{blue}{0}};
\node at (4,-0.5) {\textcolor{blue}{2}};
\node at (3,1) {\textcolor{blue}{1}};

\draw (-3,1) ellipse (0.5cm and 1cm);
\node at (-3,1) {\textcolor{red}{$Y$}};
\path[-]
(z)edge(-2.4,1)edge(-2.5,0.2)edge(-2.5,1.8)edge(-2.4,1.4)edge(-2.4,0.6);
\node at (-3,2.3) {\textcolor{blue}{2 or 3}};
\node at (0,1.5) {\textcolor{blue}{3 or 4}};
\end{tikzpicture}}

\caption{Gadget of vertex $z$ (the only which can have \b{activation} time 4). The blue numbers near the vertices show the times of an example of activation process. Vertices with time 0 belong to the target set.}
\label{gadget1b}
\end{figure} 

%\magenta{\sout{Consider initially the maximum time $k=4$.}}
We show that $\varphi$ is satisfiable if and only if $G$ contains a target set with activation time at least 4.
%\ig{being formal, $\mathcal{C}$ is the set of clauses, so it is not satisfiable. What is satisfiable is the formula $\varphi$ with variable set $\mathcal{X}$ and clause set $\mathcal{C}$ (as I wrote in the proof of \cref{thm:hard-apex}}
Suppose that $\varphi$ has a truth assignment.
For every clause $C_i$, let $k_i \in [3]$ be such that $\ell_{i,k_i}$ is set to true by the assignment.
Let $S' = \{u_{i,k_i} : i \in [m]\}$ and $S = S'\cup L$.
We show that $S$ is a target set which activates $z$ at time 4.
At time 1, $S$ activates $z_0$ and all vertices $w_{i,k_i}$ for $i \in [m]$. At time 2, $S$ activates $z_2$, $z_4$, all vertices $b_{i,k_i}$ for $i \in [m]$, and the remaining vertices of $W$.
At time 3, $S$ activates all the remaining vertices in $U$ and $B$. All vertices in $Y$ are activated by $S$ at time exactly 3, since $S$ was obtained from a truth assignment and then no vertex of $Y$ has two neighbors activated at time 2.
At time 4, $S$ activates only vertex $z$.
Thus, $G$ has activation time at least $4$.
%\ig{we have to argue that, since $\varphi$ is satisfiable, no vertex in $Y$ is activated at time 2}

Now, suppose that $t_{\tau}(G)\ge 4$ and let $S$ be a target set $S$ with activation time at least $4$.
As said before, we may assume that $S$ contains $L$. Moreover, for every clause $C_i$, $U_i\cap S \neq \emptyset$ since $|N(u_{i,p})\setminus U_i|\le 1$, for any $i \in [m]$ and $p \in [3]$. With this, we have that $S$ activates $W$ at time at most 2, $B\cup U\cup Y$ at time at most 3, and vertex $z$ at time 4. If $S$ activates a vertex of $Y$ at time 2, then $z$ is activated at time 3 (with the help of $z_0$), a contradiction. Thus no vertex of $Y$ is activated at time $2$ or less, which implies that no pair $\{u_{i,p}, u_{j,q}\}$, where $\ell_{i,p}$ is the negation of $\ell_{j,q}$, is in $S$. This means that assigning true to each $\ell_{i,p}$ for which $u_{i,p}\in S$ yields an assignment that satisfies $\varphi$.

For time values $k>4$, it suffices to add to $G$ a new path $P$ with $k-5$ edges and $k-4$ new vertices $s_1,\dotsc,s_{k-4}$, and the edge $zs_1$. Moreover, for every vertex $s_i$ of $P$, add five new vertices $s_{i,0},s_{i,1},s_{i,2},s_{i,3},s_{i,4}$ and seven edges $s_is_{i,0}, s_{i,1}s_{i,2},s_{i,3}s_{i,4}$, and $s_{i,0}s_{i,p}$ for $p \in [4]$. As before, the constructed graph $G'$ has no vertex of degree~1. From this, it is easy to see that a target set $S$ activates $s_{k-4}$ at time $k$ if and only if $S$ activates $z$ at time~4.
\end{proof}

%by reducing from the following variant of \textsc{$3$-Sat}, proved to be \NP-complete by D\'iaz et al.~\cite{DiazPSL12}.

%Let \textsc{1-Negative Planar 3-Sat} be the problem of deciding the satisfiability
%of a boolean formula $\varphi = ({\cal X},{\cal C})$, where ${\cal X}$ and ${\cal C}$ are the variables and the clauses of $\varphi$, with the following properties: (1) every variable occurs exactly once negatively and once or twice positively, (2) every clause contains two or three distinct variables, (3) every clause with three distinct variables contains at least one negative literal, and (4) the clause-variable graph $G_{\varphi}$ of $\varphi$ is planar, where $G_{\varphi} = ({\cal X} \cup {\cal C}, E)$ is the bipartite graph with vertex set ${\cal X} \cup {\cal C}$ and such that, for $x \in {\cal X}$ and $c \in {\cal C}$, $xc$ is an edge of $G_{\varphi}$ if and only if clause $c$ contains variable $x$ (either positively or negatively).} \ig{IMPORTANT: we need to allow clauses of size 2. We have to duplicate some literal in a clause of size 2}

The \emph{variable-clause incidence graph} of a \textsc{Sat} formula $\varphi = ({\cal X},{\cal C})$, where ${\cal X}$ and ${\cal C}$ are the variables and the clauses of $\varphi$, respectively, is the bipartite graph $G_{\varphi}$ with vertex set ${\cal X} \cup {\cal C}$ such that, for $x \in {\cal X}$ and $c \in {\cal C}$, $xc$ is an edge of $G_{\varphi}$ if and only if clause $c$ contains variable $x$ (either positively or negatively).

The \textsc{Restricted Planar $3$-Sat} problem is the variant of the \textsc{Sat} problem restricted to formulas $\varphi$ such that
\begin{itemize}
\item[$\bullet$] each clause has two or three literals,
\item[$\bullet$] each variable appears exactly twice positively and once negatively, and
\item[$\bullet$] the variable-clause incidence graph of $\varphi$ is planar.
\end{itemize}

Building on the proof of \cref{thm:NPC} and exploiting the fact that  \textsc{Restricted Planar $3$-Sat} is \NP-complete~\cite{DahlhausJPSY94}, we get the following result.

\begin{theorem}\label{thm:hard-apex}
The \textsc{TSS-time} problem is \NP-complete even restricted to instances $(G,\tau,k)$ such that $G$ is an apex graph, $\tau(v) = 2$ for every $v \in V(G)$, and $k \geq 4$ is fixed.
\end{theorem}
\begin{proof}
We present a polynomial reduction from the \textsc{Restricted Planar $3$-Sat}, which is \NP-complete~\cite{DahlhausJPSY94}. Given an instance $\varphi = ({\cal X},{\cal C})$ of \textsc{Restricted Planar $3$-Sat}, let~$(G, \tau, k)$ be the instance of \textsc{TSS-time} constructed in the proof of
\cref{thm:NPC} for the formula $\varphi$.
If~$C_i \in {\cal C}$ is a clause containing only two literals, we still use the same gadget depicted in \cref{gadget1a}, but removing the vertices $a_{i,3}, b_{i,3}, u_{i,3}$, and $w_{i,3}$.
%so that only vertices $w_{i,1}$ and $w_{i,2}$ (but not $w_{i,3}$) are associated with literals in $C_i$}.
%\rudC{Eu acho que isso não funcionava, pois podemos escolher o vértice $u_{i,3}$ para o target set obtendo sempre um tempo maior ou igual.}
By the proof of \cref{thm:NPC}, it follows that $t_{\tau}(G) \geq k$ if and only if $\varphi$ is satisfiable. It just remains to show that $G$ is an apex graph. More precisely, we claim that the graph obtained from $G$ be removing vertex $z$ (see \cref{gadget1b}) is planar. Clearly, it is enough to show that~$G \setminus \{z, z_0,z_1,z_2,z_3,z_4\}$ is planar.

Let $G_{\sf c}$ be the graph obtained from $G \setminus \{z, z_0,z_1,z_2,z_3,z_4\}$ by doing the following operations. First, for every $i \in [m]$, contract all the vertices in the clause gadget of $C_i$ to a single vertex. Then, for every $j \in [n]$, identify all vertices in $Y$ corresponding to a pair of occurrences of $x_j$ and $\bar{x}_j$. Since every variable appears positively and negatively in $\varphi$, it can be easily verified that $G_{\sf c}$ is isomorphic to the variable-clause incidence graph of $\varphi$. Therefore,~$G_{\sf c}$ is a planar graph. Consider an arbitrary planar embedding of $G_{\sf c}$, and we proceed to argue that it can be modified so to yield a planar embedding of $G \setminus \{z, z_0,z_1,z_2,z_3,z_4\}$. Note first that, for every $i \in [m]$, replacing the vertex in $G_{\sf c}$ corresponding to clause $C_i$ by its clause gadget in $G$ preserves planarity, since only the vertices $w_{i,1},w_{i,2},w_{i,3}$ have neighbors outside the gadget.

Now consider $j \in [n]$, and let $C_{j_1},C_{j_2},C_{j_3}$ be the three clauses of $\varphi$ containing variable~$x_j$. Equivalently, $C_{j_1},C_{j_2},C_{j_3}$ are the three neighbors of vertex $x_j$ in the graph $G_{\sf c}$. Since every variable appears twice positively and once negatively in $\varphi$, we may assume w.l.g. that $C_{j_1}$ and~$C_{j_2}$ contain $x_j$ positively, and that $C_{j_3}$ contains $x_j$ negatively. For $\ell \in [3]$, let $w_{j_{\ell},p_{\ell}}$ be the vertex in the clause gadget of $C_{j_i}$ corresponding to $x_j$ (see \cref{gadget1a}). Note that, since every variable appears exactly once negatively in $\varphi$, there are exactly two vertices in $Y \subseteq V(G)$ associated with each variable $x_j$. We add back to the planar embedding constructed so far the two vertices in $Y$ corresponding to $x_j$ as shown in \cref{fig:embedding}. Since no edge crossing is created by this construction,  we have obtained a planar embedding of~$G \setminus \{z, z_0,z_1,z_2,z_3,z_4\}$, and the theorem follows.
\end{proof}

\begin{figure}[ht]
\begin{center}
\includegraphics[scale=1]{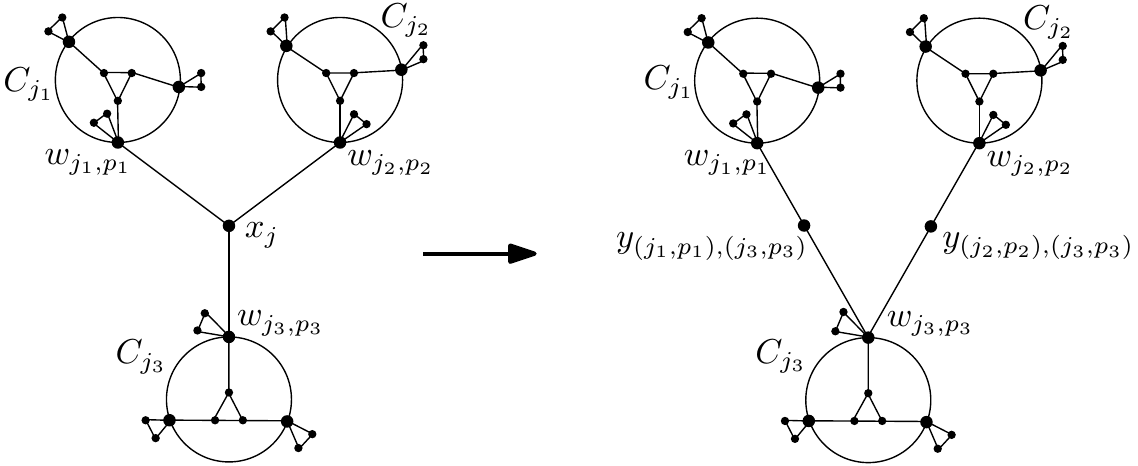}
\end{center}
\caption{Illustration of the local modification to the planar embedding of $G_{\sf c}$ so to add back the two vertices in $Y \subseteq V(G)$ associated with variable $x_j$, labeled $y_{(j_1,p_1),(j_3,p_3)}$ and $y_{(j_2,p_2),(j_3,p_3)}$  in the figure. }
%\ig{I make the figures with ipe, so this figure looks different from the other figures. I hope it is ok}
\label{fig:embedding}
\end{figure}

\b{The transformation described in the proof of \cref{thm:hard-apex} can also be applied to the original \NP-hardness proof of Benevides et al.~\cite{benevides2015} for the 2-\textsc{Neighbor Bootstrap Percolation-time} problem, and therefore we get the following corollary.}

\begin{corollary}\label{cor:P3-hard-apex}
\b{The 2-\textsc{Neighbor Bootstrap Percolation-time} problem is \NP-complete even restricted to instances $(G,k)$ such that $G$ is an apex graph and $k \geq 4$ is fixed.}
\end{corollary}

In \cite{wg2014-tcs}, it was proved that the 2-\textsc{Neighbor Bootstrap Percolation-time} problem is \NP-complete in {\sl bipartite} graphs for every fixed $k\geq 5$.
% \ig{throughout the article, we should give this problem a proper name, using \textsc{This Font}. It was not the case before, I changed it everywhere}
As before, in the \NP-hardness reduction of \cite{wg2014-tcs}, the constructed graph has many vertices of degree~1, which are not allowed in the \pname{TSS} model. In order to obtain a reduction to \pname{TSS-time} with all thresholds equal to 2, we adapt appropriately the reduction of \cite{wg2014-tcs} in order to avoid vertices of degree~1.
The solution to remove degree~1 vertices in the reduction of \cref{thm:NPC} involves many triangles of type $a_{i,j},b_{i,j},w_{i,j}$, which cannot be used here, since the graph must be bipartite. Therefore, we need to devise other gadgets.

\begin{theorem} \label{thm:NPC2}
The \textsc{TSS-time} problem is \NP-complete even restricted to instances $(G,\tau,k)$ such that $G$ is a bipartite graph, $\tau(v) = 2$ for every $v \in V(G)$, and $k \geq 5$ is fixed.
\end{theorem}

\begin{proof}
Let us prove that this restriction of the \pname{TSS-time} problem is \NP-complete by presenting, as in \cref{thm:NPC}, a polynomial reduction from the \textsc{3-Sat} problem (each clause contains exactly three literals).
%\ig{do we assume again that each clause contains exactly three literals?}

In order to simplify the reduction, let us introduce some notation.
A \emph{squared vertex} $h$ in the reduction represents the gadget of \cref{gadget2a} with auxiliary vertices $h_1,h_2,h_{11},\ldots,h_{26}$ all whose edges are represented in the figure. We assume that all thresholds are equal to 2.
One important fact about squared vertices is that, for any target set, its activation time is always at most 2. This is because any target set must contain at least a vertex of each one of the following sets: $\{h_{11},h_{14}\}$, $\{h_{12},h_{15}\}$, $\{h_{13},h_{16}\}$, $\{h_{21},h_{24}\}$, $\{h_{22},h_{25}\}$, and $\{h_{23},h_{26}\}$ (otherwise the set is never activated). If $h$ is not activated at time 0 or 1, then by the pigeonhole principle any target set contains at least two vertices in each one of $\{h_{14},h_{15},h_{16}\}$ and $\{h_{24},h_{25},h_{26}\}$. From this, $h_1$ and $h_2$ are activated at time 1 and consequently $h$ is activated at time~2.

%\vspace{-.5cm}
\begin{figure}[ht]
\centering\scalebox{0.75}{
\begin{tikzpicture}[scale=1]
\tikzstyle{vertex}=[draw,circle,fill=white,minimum size=15pt,inner sep=2pt]
\tikzstyle{vertex1}=[draw,circle,fill=black!5,minimum size=25pt,inner sep=4pt]
\tikzstyle{vertex2}=[draw,fill=black!5,minimum size=25pt,inner sep=4pt]

\node[vertex2] (hh)   at (-5,3) {$h$};

\path[-,thick] (-3,2.7) edge (-2,2.7);
\path[-,thick] (-3,3.0) edge (-2,3.0);
\path[-,thick] (-3,3.3) edge (-2,3.3);

\node[vertex1] (h)   at (3,3) {$h$};
\node[vertex] (h2)  at (6,3) {$h_2$};
\node[vertex] (h23) at (4,2) {$h_{23}$};
\node[vertex] (h22) at (4,3) {$h_{22}$};
\node[vertex] (h21) at (4,4) {$h_{21}$};
\node[vertex] (h26) at (5,2) {$h_{26}$};
\node[vertex] (h25) at (5,3) {$h_{25}$};
\node[vertex] (h24) at (5,4) {$h_{24}$};
\node[vertex] (h1)  at (0,3) {$h_1$};
\node[vertex] (h13) at (2,2) {$h_{13}$};
\node[vertex] (h12) at (2,3) {$h_{12}$};
\node[vertex] (h11) at (2,4) {$h_{11}$};
\node[vertex] (h16) at (1,2) {$h_{16}$};
\node[vertex] (h15) at (1,3) {$h_{15}$};
\node[vertex] (h14) at (1,4) {$h_{14}$};

\draw [-] (h) to [bend right=80,looseness=2] (h1);
\draw [-] (h) to [bend  left=80,looseness=2] (h2);

\path[-,thick]
(h) edge (h11) edge (h12) edge (h13)
(h11) edge (h14) (h12) edge (h15) (h13) edge (h16)
(h1) edge (h14) edge (h15) edge (h16);

\path[-,thick]
(h) edge (h21) edge (h22) edge (h23)
(h21) edge (h24) (h22) edge (h25) (h23) edge (h26)
(h2) edge (h24) edge (h25) edge (h26);
\end{tikzpicture}}

\caption{Gadget of a squared vertex $h$, which is always activated at time at most $2$.}
\label{gadget2a}
\end{figure}
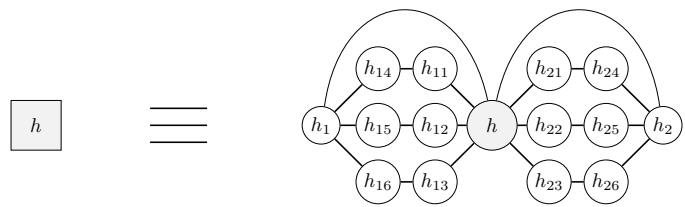 

A \emph{double squared vertex} $h$ in the reduction represents the gadget of \cref{gadget2b} with auxiliary vertices $h',h_0,h_{01},\dotsc,h_{06}$ all whose edges are represented in the figure (notice that $h'$ is squared and also contains the edges in \cref{gadget2a}).

One important fact about double squared vertices is that, for any target set, its activation time is always at most 3. This is because, with identical arguments as before, if $h$ is not activated at time 0 and 1, $h_0$ is activated at time~1 and then $h$ is activated at time at most~3, since $h'$ is squared and is activated at time at most 2. Another important fact is that, if a neighbor of $h$ outside the gadget of \cref{gadget2b} is activated at time 0 or 1, then $h$ is activated at time at most 2 (since $h_0$ and this neighbor activate $h$).

\begin{figure}[ht]
\centering\scalebox{0.75}{
\begin{tikzpicture}[scale=1]
\tikzstyle{vertex}=[draw,circle,fill=white,minimum size=15pt,inner sep=2pt]
\tikzstyle{vertex1}=[draw,circle,fill=black!20,minimum size=25pt,inner sep=4pt]
\tikzstyle{vertex2}=[draw,fill=black!5,minimum size=25pt,inner sep=4pt]
\tikzstyle{vertex3}=[draw,double,fill=black!20,minimum size=25pt,inner sep=4pt]

\node[vertex3] (hh)   at (-5,3) {$h$};

\path[-,thick] (-3,2.7) edge (-2,2.7);
\path[-,thick] (-3,3.0) edge (-2,3.0);
\path[-,thick] (-3,3.3) edge (-2,3.3);

\node[vertex1] (h)  at (3,3) {$h$};
\node[vertex2] (h2)  at (6,3) {$h'$};
\node[vertex] (h1)  at (0,3) {$h_0$};
\node[vertex] (h13) at (2,2) {$h_{03}$};
\node[vertex] (h12) at (2,3) {$h_{02}$};
\node[vertex] (h11) at (2,4) {$h_{01}$};
\node[vertex] (h16) at (1,2) {$h_{06}$};
\node[vertex] (h15) at (1,3) {$h_{05}$};
\node[vertex] (h14) at (1,4) {$h_{04}$};

\draw [-] (h) to [bend right=80,looseness=2] (h1);

\path[-,thick]
(h) edge (h2) edge (h11) edge (h12) edge (h13)
(h11) edge (h14) (h12) edge (h15) (h13) edge (h16)
(h1) edge (h14) edge (h15) edge (h16);

\end{tikzpicture}}

\caption{Gadget of double squared vertex $h$, which is always activated at time at most 3. Moreover, if a neighbor of $h$ outside this gadget is activated at time 0 or 1, $h$ is activated at time at most 2.}
\label{gadget2b}
\end{figure}
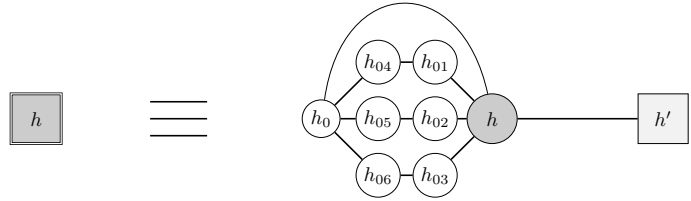 

Let $\varphi = ({\cal X},{\cal C})$ be an instance of \pname{$3$-SAT}, where
${\cal X}=\{x_1, \dotsc,x_n\}$ is the set of variables and $\mathcal{C}=\{C_1,\dotsc,C_m\}$ is the set of clauses.
Let us denote the three literals of $C_i$ by $\ell_{i,1}$, $\ell_{i,2}$ and $\ell_{i,3}$.
%\ig{so it seems that we are indeed assuming the each clause has exactly 3 literals}
We proceed by constructing a graph $G$ such that $t_{\tau}(G)\geq 5$ if and only if the \textsc{$3$-Sat} instance is satisfiable.

For every clause $C_i$ of $\mathcal{C}$, add to $G$ the gadget of \cref{gadget-bip-t5}.
Let $W = \{w_{i,p} \mid i \in [m], p \in [3]\}$.
For every pair of complement literals $\ell_{i,p}, \ell_{j,q}$,
%(that is, such that $\ell_{i,p} = \bar{\ell}_{j,q}$),
add a vertex~$y_{(i,p),(j,q)}$ adjacent to $w_{i,p}$ and $w_{j,q}$. Let $Y$ be the set of all vertices $y_{(i,p),(j,q)}$. Finally, add a vertex $z$ adjacent to all vertices in $Y$ and a squared vertex $z'$ adjacent to $z$.

Notice that $G$ has no vertex of degree~1, as required in the definition of threshold function. To prove that $G$ is bipartite, consider the following partition $(A,B)$ of the main vertices of $G$. $A$ contains all vertices $u^A_{i,j},a_{i,j},h_{i,j},y_{(i,p),(j,q)}$, and $z'$.
$B$ contains all vertices $u^B_{i,j},w_{i,j},b_{i,j},c_{i,j}$, and $z$.
Moreover, the gadgets of \cref{gadget2a} and \cref{gadget2b} are clearly bipartite.

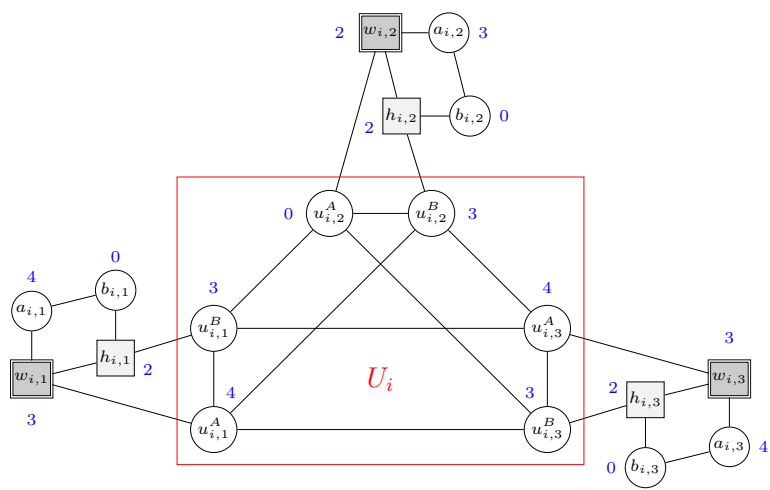
\begin{figure}[ht]
\centering\scalebox{0.9}{
\begin{tikzpicture}[scale=.85]
\tikzstyle{every node}=[font=\scriptsize]
%\tikzstyle{vertex}=[draw,circle,fill=white,minimum size=14pt,inner sep=1pt]
\tikzstyle{vertex}=[draw,circle,fill=white,minimum size=15pt,inner sep=1pt]
\tikzstyle{vertex2}=[draw,regular polygon sides=4,fill=black!5,minimum size=15pt,inner sep=1pt]
\tikzstyle{vertex3}=[draw,regular polygon sides=4,double,fill=black!20,minimum size=15pt,inner sep=1pt]

\def \diff {17}
\def \dist {3cm}
\foreach \i/\im [evaluate=\i as \angle using (\i*90)-90] in {1/3,2/2,3/1}{
	\node[vertex] (uB\i) at (\angle - \diff:\dist) {$u_{i,\im}^B$};
    \node[vertex] (uA\i) at (\angle + \diff:\dist) {$u_{i,\im}^A$};
    \node[vertex3] (w\i)  at (\angle:\dist+3cm) {$w_{i,\im}$};
    \node[vertex2,shift={(90*\i-90 + \diff:1.5cm)}] (h\i) at (uB\i) {$h_{i,\im}$};
    \node[vertex,shift={(90*\i-180:1cm)}] (a\i) at (w\i) {$a_{i,\im}$};
    \node[vertex,shift={(90*\i-180:1cm)}] (b\i) at (h\i) {$b_{i,\im}$};
    \draw[-] (uA\i) to (uB\i);
    \draw[-] (uA\i) to (w\i);
    \draw[-] (uB\i) to (h\i);
    \draw[-] (h\i) to (w\i);
    \draw[-] (h\i) to (b\i);
    \draw[-] (a\i) to (b\i);
    \draw[-] (a\i) to[bend right=0] (w\i);
}
\draw[-] (uA1) to (uB2);
\draw[-] (uA1) to (uB3);
\draw[-] (uA2) to (uB1);
\draw[-] (uA2) to (uB3);
\draw[-] (uA3) to (uB1);
\draw[-] (uA3) to (uB2);
\draw[red] (135:4.95cm) rectangle (-\diff-6:3.8cm);
\node[shift={(90 :0.6cm)}] at (w1) {\textcolor{blue}{3}};
\node[shift={(180:0.6cm)}] at (w2) {\textcolor{blue}{2}};
\node[shift={(270:0.6cm)}] at (w3) {\textcolor{blue}{3}};
\node[shift={(160:0.5cm)}] at (h1) {\textcolor{blue}{2}};
\node[shift={(200:0.5cm)}] at (h2) {\textcolor{blue}{2}};
\node[shift={(-20:0.5cm)}] at (h3) {\textcolor{blue}{2}};
\node[shift={( 0 :0.5cm)}] at (a1) {\textcolor{blue}{4}};
\node[shift={(180:0.5cm)}] at (b1) {\textcolor{blue}{0}};
\node[shift={( 0 :0.5cm)}] at (a2) {\textcolor{blue}{3}};
\node[shift={( 0 :0.5cm)}] at (b2) {\textcolor{blue}{0}};
\node[shift={(90 :0.5cm)}] at (a3) {\textcolor{blue}{4}};
\node[shift={(90 :0.5cm)}] at (b3) {\textcolor{blue}{0}};
\node[shift={(90 :0.6cm)}] at (uA1){\textcolor{blue}{4}};
\node[shift={(115:0.6cm)}] at (uB1){\textcolor{blue}{3}};
\node[shift={(180:0.6cm)}] at (uA2){\textcolor{blue}{0}};
\node[shift={( 0 :0.6cm)}] at (uB2){\textcolor{blue}{3}};
\node[shift={(65 :0.6cm)}] at (uA3){\textcolor{blue}{4}};
\node[shift={(90 :0.6cm)}] at (uB3){\textcolor{blue}{3}};
\node at (0,0){\textcolor{red}{\Large $U_i$}};
\end{tikzpicture}}
\caption{Bipartite gadget of clause $C_i$. The blue numbers near the vertices are the activation times when $u^A_{i,2}$ is in the target set. The set $U_i$ is inside the red square.} \label{gadget-bip-t5}
\end{figure} 

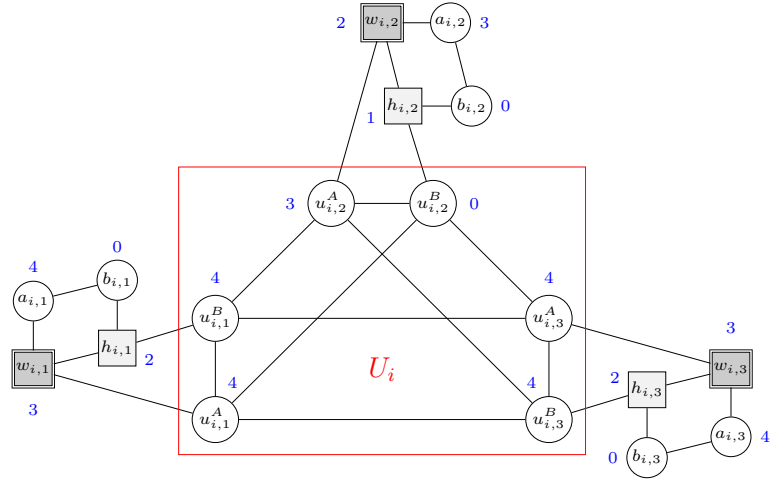
\begin{figure}[ht]
\centering\scalebox{0.9}{
\begin{tikzpicture}[scale=.85]
\tikzstyle{every node}=[font=\scriptsize]
%\tikzstyle{vertex}=[draw,circle,fill=white,minimum size=14pt,inner sep=1pt]
\tikzstyle{vertex}=[draw,circle,fill=white,minimum size=15pt,inner sep=1pt]
\tikzstyle{vertex2}=[draw,regular polygon sides=4,fill=black!5,minimum size=15pt,inner sep=1pt]
\tikzstyle{vertex3}=[draw,regular polygon sides=4,double,fill=black!20,minimum size=15pt,inner sep=1pt]

\def \diff {17}
\def \dist {3cm}
\foreach \i/\im [evaluate=\i as \angle using (\i*90)-90] in {1/3,2/2,3/1}{
	\node[vertex] (uB\i) at (\angle - \diff:\dist) {$u_{i,\im}^B$};
    \node[vertex] (uA\i) at (\angle + \diff:\dist) {$u_{i,\im}^A$};
    \node[vertex3] (w\i)  at (\angle:\dist+3cm) {$w_{i,\im}$};
    \node[vertex2,shift={(90*\i-90 + \diff:1.5cm)}] (h\i) at (uB\i) {$h_{i,\im}$};
    \node[vertex,shift={(90*\i-180:1cm)}] (a\i) at (w\i) {$a_{i,\im}$};
    \node[vertex,shift={(90*\i-180:1cm)}] (b\i) at (h\i) {$b_{i,\im}$};
    \draw[-] (uA\i) to (uB\i);
    \draw[-] (uA\i) to (w\i);
    \draw[-] (uB\i) to (h\i);
    \draw[-] (h\i) to (w\i);
    \draw[-] (h\i) to (b\i);
    \draw[-] (a\i) to (b\i);
    \draw[-] (a\i) to[bend right=0] (w\i);
}
\draw[-] (uA1) to (uB2);
\draw[-] (uA1) to (uB3);
\draw[-] (uA2) to (uB1);
\draw[-] (uA2) to (uB3);
\draw[-] (uA3) to (uB1);
\draw[-] (uA3) to (uB2);
\draw[red] (135:4.95cm) rectangle (-\diff-6:3.8cm);
\node[shift={(90 :0.6cm)}] at (w1) {\textcolor{blue}{3}};
\node[shift={(180:0.6cm)}] at (w2) {\textcolor{blue}{2}};
\node[shift={(270:0.6cm)}] at (w3) {\textcolor{blue}{3}};
\node[shift={(160:0.5cm)}] at (h1) {\textcolor{blue}{2}};
\node[shift={(200:0.5cm)}] at (h2) {\textcolor{blue}{1}};
\node[shift={(-20:0.5cm)}] at (h3) {\textcolor{blue}{2}};
\node[shift={( 0 :0.5cm)}] at (a1) {\textcolor{blue}{4}};
\node[shift={(180:0.5cm)}] at (b1) {\textcolor{blue}{0}};
\node[shift={( 0 :0.5cm)}] at (a2) {\textcolor{blue}{3}};
\node[shift={( 0 :0.5cm)}] at (b2) {\textcolor{blue}{0}};
\node[shift={(90 :0.5cm)}] at (a3) {\textcolor{blue}{4}};
\node[shift={(90 :0.5cm)}] at (b3) {\textcolor{blue}{0}};
\node[shift={(90 :0.6cm)}] at (uA1){\textcolor{blue}{4}};
\node[shift={(115:0.6cm)}] at (uB1){\textcolor{blue}{4}};
\node[shift={(180:0.6cm)}] at (uA2){\textcolor{blue}{3}};
\node[shift={( 0 :0.6cm)}] at (uB2){\textcolor{blue}{0}};
\node[shift={(65 :0.6cm)}] at (uA3){\textcolor{blue}{4}};
\node[shift={(90 :0.6cm)}] at (uB3){\textcolor{blue}{4}};
\node at (0,0){\textcolor{red}{\Large $U_i$}};
\end{tikzpicture}}
\caption{Bipartite gadget of clause $C_i$. The blue numbers near the vertices are the activation times when $u^B_{i,2}$ is in the target set. The set $U_i$ is inside the red square.} \label{gadget-bip-t5B}
\end{figure} 

Suppose that $\varphi$ has a truth assignment. For every clause $C_i$, let $k_i\in [3]$ such that $\ell_{i,k_i}$ is true. We obtain a target set $S$ of $G$ as follows: $S$ contains $u^A_{i,k_i}$, $b_{i,1}$, $b_{i,2}$, and $b_{i,3}$ for every clause $C_i$. Moreover, $S$ contains the vertices $h_{j,4}$, $h_{j,5}$, and $h_{j,6}$ in \cref{gadget2a} and \cref{gadget2b} for every squared or double squared vertex.
It is not difficult to see in \cref{gadget-bip-t5} that all vertices in the clause gadgets are activated by $S$ at time at most 4. Also, for every clause $C_i$, we have that $S$ activates $w_{i,k_i}$  at time 2 and activates $w_{i,k'}$ at time 3 for $k' \in [3] \setminus \{k_i\}$. From the truth assignment, all vertices of $Y$ are activated at time exactly 4, since every vertex $y\in Y$ is adjacent to exactly one vertex in $W$ activated at time 2 and to another vertex in $W$ activated at time 3. Thus,  vertex $z$ is activated at time 5 and consequently $t_{\tau}(G)\geq 5$.

Now, suppose that $t_{\tau}(G)\geq 5$ and let $S$ be a target set with activation time at least 5.
%\ig{if we want to say \emph{exactly} 5, then we have to use \cref{lem:equal-at-least}}
Then $S$ has at least one vertex of every set $U_i$, since any vertex in $U_i$ has only one neighbor outside $U_i$. Moreover, from the same argument, $S$ contains $a_{i,p}$ or $b_{i,p}$ for every $p \in [3]$.
It is not difficult to see in \cref{gadget-bip-t5} and \cref{gadget-bip-t5B} that all vertices in the clause gadgets are activated by $S$ at time at most 4 and all vertices of $W$ are activated at time at most 3 (in the figures, $S$ is represented by the vertices with activation time 0).
We may assume w.l.g. that $S$ contains $b_{i,p}$ instead of $a_{i,p}$, for $p \in [3]$, since $a_{i,p}$ activates $w_{i,p}$ at time at most 2.
Also, all vertices of $Y$ are activated at time at most 4. Also recall that $z'$ is activated at time at most 2. Therefore, $z$ is the unique vertex activated by $S$ at time 5 and consequently all vertices of $Y$ must be activated at time exactly~4.

With this, consider the following assignment. For every $u_{i,p}$ in $S$, for $p \in [3]$, assign true to the literal $\ell_{i,p}$. That is, if the literal $\ell_{i,p}$ is a positive literal, assign true to its variable; otherwise, assign false to its variable.
This is a valid truth assignment, since all vertices of $Y$ are activated at time 4 and consequently any two vertices of $U$ representing complementary literals cannot be both in $S$.
Moreover, this assignment satisfies all the clauses in $\mathcal{C}$, since $S$ has at least one vertex in each set $U_i$ and consequently the assignment satisfies at least one literal of every clause.

%\c{For time values $k>5$, it suffices to alter the reduction by adding a path $P$ of length $k-5$ and linking one end of $P$ to the vertex $z$, appending a new leaf vertex to each vertex in $P$. The proof remains  the same.} \ig{This paragraph should be deleted, right?}

For time values $k>5$, it suffices to include a new path $P$ with $k-6$ edges and $k-5$ new vertices $s_1,\dotsc,s_{k-5}$ and to add the edge $zs_1$. Moreover create $k-5$ new squared vertices $s'_i$ and add the edge $s_is_i'$ for $i \in [k-5]$. As before, $G$ has no vertex of degree~1. From this, it is easy to see that a target set $S$ activates $s_{k-5}$ at time at least~$k$ if and only if $S$ activates $z$ at time at least~5.
\end{proof}

\section{TSS-time is FPT in graphs of bounded local treewidth}\label{sec:FPT-Alg}
%\ig{This section is new}

In this section we provide an \FPT algorithm to solve the \textsc{GTSS-time} problem (so the \textsc{TSS-time} problem as well) in graphs of bounded local treewidth (\cref{thm:FPT}). This result together the \NP-completeness result of \cref{thm:hard-apex} will yield the complexity dichotomy proved in \cref{thm:dichotomy}.
We first need to introduce some notation and a slightly (more) generalized version of the \textsc{GTSS-time} problem.

Let $G$ be a graph and $\tau$ be a \b{generalized} threshold function in $G$. Recall that {$\tau^*= \max_{v \in V(G)}\tau(v)$}. Given a subset $V_{\sf f} \subseteq V(G)$ of \emph{forced} vertices, we denote by $t_{\tau}(G,V_{\sf f})$ the maximum activation time $t_{\tau}(S_0)$ among all target sets $S_0$ of $G$ such that~$V_{\sf f} \subseteq S_0$.
Clearly, $t_{\tau}(G,\emptyset) = t_{\tau}(G)$, hence deciding whether $t_{\tau}(G,V_{\sf f}) \geq k$, for a positive integer $k$, is equivalent to the \pname{GTSS-time} problem if we consider the threshold of any vertex in $V_{\sf f}$ strictly greater than its degree, while the threshold of any other vertex is maintained. \b{However, we still need this generalized version of the problem in this section, since we do not want the value of $\tau^*$ to increase when considering the auxiliary subproblems that we will define below.}

Before this, we show in the next lemma that deciding whether there exists a target set with activation time {\sl at least} $k$ is equivalent to the {\sl exact} version.

\begin{lemma}\label{lem:equal-at-least}
For every graph $G$, every generalized threshold function $\tau$ in $G$, every set $V_{\sf f} \subseteq V(G)$, and every positive integer $k$, $t_{\tau}(G,V_{\sf f}) \geq k$ if and only if there is target set $S_0 \subseteq V(G)$ with~$V_{\sf f} \subseteq S_0$ and such that  $t_{\tau}(S_0) = k$. In particular, $t_{\tau}(G) \geq k$ if and only if there is target set $S_0 \subseteq V(G)$ such that  $t_{\tau}(S_0) = k$.
\end{lemma}

\begin{proof}
If there is target set $S_0 \subseteq V(G)$ with $V_{\sf f} \subseteq S_0$ and such that  $t_{\tau}(S_0) = k$, then by definition $t_{\tau}(G, V_{\sf f}) \geq k$, so let us focus on the forward direction. Let $S_0$ be a target set of $G$ with $V_{\sf f} \subseteq S_0$ and such that $t_{\tau}(S_0) \geq k$, and let $S_0, S_1, S_2, \ldots, S_t$ be the partition of  $V(G)$ given by  $S_0$, where $t \geq k$. If $t=k$, then we are done. Otherwise, let $S_0' = S_0 \cup S_1 \cup \dots \cup S_{t-k}$. It can be easily verified that $S_0'$ is a target set of $G$ with $V_{\sf f} \subseteq S_0'$ and such that~$t_{\tau}(S_0') =~k$.
\end{proof}

%Let $G$ be a graph, let $v \in V(G)$, and let $p$ be a positive integer. By \emph{attaching $p$ triangles to $v$} we mean the operation of adding $p$ vertex-disjoint triangles to $G$ and identifying one vertex of each triangle with $v$. We call these triangles \emph{pendent triangles}.

We now define auxiliary graphs that will be used in the \FPT algorithm of \cref{thm:FPT}. The crucial property of these auxiliary graphs is that they have diameter $\Ocal(k)$, which will be exploited in order to bound their treewidth.  Let $G$ be a graph, let $\tau$ be a generalized threshold function in $G$, let $v \in V(G)$, and let $k$ be a positive integer. We define the pair $(G_{k}^v, \tau_k^v)$ such that~$G_{k}^v$ is a graph and $\tau_k^v$ is a generalized threshold function in $G_{k}^v$, as follows. Let $G_{k}^v = G[N_k[v]]$, that is, the subgraph of $G$ induced by the vertices at distance at most $k$ from $v$ in $G$ (including $v$), and let
%$\karol{I would highlight that $v$ is included, but it is not necessary}, 
%(that is, the induced subgraph of $G$ containing all vertices within distance at most $k$ from $v$)
%by attaching $\tau(u)$ triangles to every vertex $u$ that is at distance exactly $k$ from $v$ in $G$. We define $\tau_k^v:V(G_{k}^v)\to\NN$ so that
$\tau_k^v$ be the restriction of $\tau$ to $N_k[v]$. \b{Note that $\tau_k^v$ may not be a threshold function in $G_{k}^v$, even if $\tau$ is a threshold function in $G$, since the degree in $G_{k}^v$ of the vertices at distance exactly $k$ from $v$ in $G$ may have decreased, but $\tau_k^v$ is still a {\sl generalized} threshold function in $G_{k}^v$ such that $(\tau_k^v)^{*} \leq \tau^*$.}
%defined such that, if $u \in N_{k-1}[v]$ then $\tau_k^v(u) = \tau(u)$, and  $\tau_k^v(u) = 1$ otherwise. Since every vertex at distance exactly $k$ from $v$ in $G$ has some neighbor in $N_{k-1}[v]$, it can be easily verified that $\tau_k^v$ is indeed a threshold function in $G_{k}^v$, and this is the only reason for which we have set $\tau_k^v(u) = 1$ for every vertex $u$ at distance exactly $k$ from $v$ in $G$. Nevertheless, the value of $\tau_k^v$ for those vertices will be irrelevant since, as stated in the next lemma, we will always include them in the set $V_{\sf f}$ of forced vertices that must belong to any target set.
 In the next lemma we show that dealing with the auxiliary graphs $G_{k}^v$ is enough in order to solve the \textsc{GTSS-time} problem.

\begin{lemma}\label{lem:balls}
For every graph $G$, every generalized threshold function $\tau$ in $G$, and every positive integer~$k$, $t_{\tau}(G) \geq k$ if and only if there exists a vertex $v \in V(G)$ such that $t_{\tau_k^v}(G_{k}^v, V_{\sf f}) \geq k$, where $V_{\sf f}$ is the set of vertices at distance exactly $k$ from $v$ in~$G$.
\end{lemma}
\begin{proof} Suppose first that $t_{\tau}(G) \geq k$. By \cref{lem:equal-at-least}, there is target set $S_0 \subseteq V(G)$ such that~$t_{\tau}(S_0) = k$. Let $V(G) = S_0 \uplus S_1 \uplus \dots \uplus S_k$ be the partition of $V(G)$ into $k$ non-empty sets given by \magenta{the activation process starting at} $S_0$,  let $v$ be any vertex in $S_k$, and let $V_{\sf f}$ be  the set of vertices at distance exactly~$k$ from $v$ in~$G$.  We claim that $t_{\tau_k^v}(G_{k}^v,V_{\sf f} ) \geq k$. Let $S_0^v = (S_0 \cap N_{k-1}(v)) \cup V_{\sf f}$. Since~$V_{\sf f} \subseteq S_0^v$, we just have to verify that $S_0^v$ is a target set of $G_{k}^v$ with $t_{\tau_k^v}(S_0^v) \geq k$. Since $S_0$ activates vertex $v$ at time exactly $k$ in $G$, $S_0 \cup \{u\}$ also activates vertex $v$ at time exactly $k$ for any vertex $u$ at distance exactly $k$ from $v$ in $G$. Iterating this argument, it follows that $S_0 \cup V_{\sf f}$ activates vertex $v$ at time exactly $k$ in $G$. Thus, $S_0^v$ activates vertex $v$ at time exactly $k$ in~$G_{k}^v$. As for the other vertices of $G_{k}^v$, since $S_0$ is a target set of $G$ and $V_{\sf f} \subseteq S_0^v$, it follows that $S_0^v$ is indeed a target set of $G_{k}^v$ containing $V_{\sf f}$ that activates $v$ at time exactly $k$, and therefore $t_{\tau_k^v}(G_{k}^v, V_{\sf f}) \geq k$.

Conversely, suppose that there exists $v \in V(G)$ such that $t_{\tau_k^v}(G_{k}^v, V_{\sf f}) \geq k$, where $V_{\sf f}$ is the set of vertices at distance exactly $k$ from $v$ in~$G$. Let $S_0^v$ be a target set of $G_{k}^v$ containing~$V_{\sf f}$ such that $t_{\tau_k^v}(S_0^v) \geq k$. Let $S_0 = S_0^v \cup (V(G) \setminus N_{k}[v])$. That is, $S_0$ contains $S_0^v$ and all vertices at distance at least $k+1$ from $v$ in $G$. Since $V_{\sf f} \subseteq S_0^v$, $S_0$ also contains all vertices at distance exactly $k$ from $v$ in $G$. We claim that $S_0$ is a target set of $G$ with $t_{\tau}(S_0) \geq k$. The fact that~$S_0$ is a target set of $G$ follows from the hypothesis that $S_0^v$ is a target set of $G_{k}^v$ and the fact that $S_0$ contains all vertices at distance at least $k$ from $v$ in $G$. On the other hand, since~$t_{\tau_k^v}(S_0^v) \geq k$, $V_{\sf f} \subseteq S_0^v \subseteq S_0$, and no vertex in $V(G) \setminus N_{k}[v]$ has a neighbor in the set~$N_{k-1}[v]$, it follows that $t_{\tau}(S_0) \geq k$, and therefore $t_{\tau}(G) \geq k$.
\end{proof}

The last ingredient that we need before proving \cref{thm:FPT} is to show that deciding whether $t_{\tau}(G,V_{\sf f}) \geq k$ can be expressed by an \MSOone formula of appropriate length. Note that, in particular, this applies to deciding whether $t_{\tau}(G,\emptyset) = t_{\tau}(G)\geq k$, that is, to the \textsc{GTSS-time} problem.

\begin{lemma}\label{lem:MSOL}
 Given a graph $G$, a generalized threshold function $\tau$ in $G$, a subset $V_{\sf f} \subseteq V(G)$, and a positive integer $k$, the problem of deciding whether $t_{\tau}(G,V_{\sf f}) \geq k$ can be expressed by an~\MSOone formula $\phi$ whose length depends on $k$ and $\tau^*$.
%The \textsc{TSS-time} problem, with instances of the form $(G, \tau, k)$, can be expressed by an \MSOone formula $\phi$ whose length depends on $k$ and $\max_{v \in V(G)}\tau(v)$.
\end{lemma}
\begin{proof}
We may assume that $G$ is given along with a partition of $V(G)$ according to $V_{\sf f}$ and the values of the generalized threshold function $\tau$, namely $V(G) = V_{\sf f} \uplus V_0 \uplus V_1 \uplus \dots \cup V_{\tau^{\star}}$, where~$\tau^{\star} = \max_{v \in V(G)}\tau(v)$ and, for $j \in [0,\tau^{\star}]$, $V_j = \{v \in V(G) \setminus V_{\sf f} \mid \tau(v) = j\}$. Note that $V_{\sf f}$ and the sets $V_j$ may be empty. By \cref{lem:equal-at-least}, $(G, \tau, k)$ is a \yes-instance of \textsc{GTSS-time} %\karol{Não seria \textsc{GTSS-time}?} 
if and only there exists a target set $S_0 \subseteq V(G)$ such that $t_{\tau}(G,\tau,S_0) = k$. The existence of such a set~$S_0$ can be equivalently expressed as the existence of
 a partition $V(G)= S_0 \uplus S_1 \uplus \dots \uplus S_k$ into~$k+1$ non-empty sets  with $V_{\sf f} \subseteq S_0$ and such that
 \begin{enumerate}
\item[(i)] for every $i \in [2,k]$ and every $j \in  [0,\tau^{\star}]$, every vertex in $S_i \cap V_j$ has strictly less than $j$ neighbors in the set $\bigcup_{h=0}^{i-2}S_h$, and
\item[(ii)] for every $i \in [1,k]$ and every $j \in  [0,\tau^{\star}]$, every vertex in $S_i \cap V_j$ has at least $j$ neighbors in the set $\bigcup_{h=0 }^{i-1}S_h$.
\end{enumerate}
Let us argue that the above conditions can be indeed expressed by an \MSOone formula $\phi$ whose length depends only on $k$ and $\tau^*$. First, note that the existence of $k+1$ pairwise disjoint vertex sets that form a partition of $V(G)$ with $V_{\sf f} \subseteq S_0$ can be easily expressed in~\MSOone. On the other hand, in order to express condition~(i) above, it suffices to quantify, for every~$i \in [2,k]$, every $j \in  [0,\tau^{\star}]$, and every vertex $v \in S_i \cap V_j$, the non-existence of $j$ neighbors of $v$ in the set $\bigcup_{h=0}^{i-2}S_h$. Finally, as for condition~(ii), it suffices to quantify, for every $i \in [1,k]$, every $j \in  [0,\tau^{\star}]$, and every vertex $v \in S_i \cap V_j$, the existence of $j$ neighbors of~$v$ in the set $\bigcup_{h=0}^{i-1}S_h$. Clearly, the length of the obtained \MSOone formula $\phi$ is bounded by a function of $k$ and $\tau^*$, and the lemma follows.
\end{proof}

%\ig{Note that \cref{thm:NPC} proves that the \textsc{TSS-time} problem is \NP-complete even for fixed $k=4$ and $ \tau^{*} = 2$}
We finally have all the ingredients to prove our \FPT algorithm to solve the \textsc{GTSS-time} problem. For a graph class ${\cal C}$, we denote by \textsc{GTSS-time}$|_{{\cal C}}$ (resp. \textsc{TSS-time}$|_{{\cal C}}$) the restriction of the \textsc{GTSS-time} (resp. \textsc{TSS-time}) problem to input graphs $G$ belonging to ${\cal C}$.

\begin{theorem}\label{thm:FPT}
If ${\cal C}$ is a graph class of bounded local treewidth, then the \textsc{GTSS-time}$|_{{\cal C}}$ problem is \FPT parameterized by $k$ and $\tau^*$.
\end{theorem}
\begin{proof} Let $(G, \tau, k)$ be an instance of \textsc{GTSS-time} where $G \in {\cal C}$, $\tau$ is a generalized threshold function in $G$, and $k$ is a positive integer. Since ${\cal C}$ has bounded local treewidth, there exists a function~$f:\NN \to\NN$  such that, for every graph $G\in{\cal C}$, every vertex $v \in V(G)$ and every positive integer $r$, $\tw(G[N_r[v]])\leq f(r)$.

By \cref{lem:balls}, $t_{\tau}(G) \geq k$ if and only if there exists a vertex $v \in V(G)$ such that~$t_{\tau_k^v}(G_{k}^v, V_{\sf f}) \geq k$, where $V_{\sf f}$ is the set of vertices at distance exactly $k$ from $v$ in~$G$. Based on this, for every vertex $v \in V(G)$, we generate in linear time the graph $G_{k}^v$, and it is enough to decide whether $t_{\tau_k^v}(G_{k}^v, V_{\sf f}) \geq k$. Since $G_{k}^v = G[N_k[v]]$ and $G \in {\cal C}$, we have that~$\tw(G_{k}^v) \leq f(k)$.
By \cref{lem:MSOL}, deciding whether $t_{\tau_k^v}(G_{k}^v, V_{\sf f}) \geq k$ can be expressed by an \MSOone (in particular, \MSO) formula $\phi$ whose length depends only on $k$ and~$(\tau_k^v)^* \leq \tau^*$. Therefore,
 \cref{thm:Courcelle} implies that deciding whether $t_{\tau_k^v}(G_{k}^v, V_{\sf f}) \geq k$ can be solved in time~$g(k, \tau^*, \tw(G_{k}^v)) \cdot n$ for some computable function $g$, where $n = |V(G)|$. Since
$\tw(G_{k}^v) \leq f(k)$, deciding whether~$t_{\tau}(G) \geq k$ can be solved in time~$h(k, \tau^*) \cdot n^2$ for some computable function~$h:\Bbb{N}^{2}\to\Bbb{N}$, and the theorem follows.
\end{proof}

As particular cases of \cref{thm:FPT}, it follows that the \textsc{GTSS-time}$|_{{\cal C}}$ problem is \FPT parameterized by $k$ and $\tau^*$ when ${\cal C}$ is the class of graphs of treewidth bounded by a constant, the class of graphs of maximum degree bounded by a constant, the class of planar graphs or, more generally, the class of graphs embeddable in a fixed surface (i.e., graphs of bounded genus).

%\begin{corollary}\label{cor:FPT}
%Diameter-treewidth property, planar, bounded genus, bounded treewidth...
%\qed\end{corollary}

With \cref{thm:hard-apex} and \cref{thm:FPT} at hand, the following theorem can be easily proved.

\begin{theorem}\label{thm:dichotomy}
Let ${\cal C}$ be a minor-closed graph class.  Then  \textsc{TSS-time}$|_{{\cal C}}$ is% \vspace{-.25cm}
\begin{itemize}
\item[$\bullet$] \FPT parameterized by $k$ and $\tau^*$, if ${\cal C}$ has bounded local treewidth.
\item[$\bullet$] \NP-complete for every fixed $k \geq 4$ and $ \tau^{*} = 2$, otherwise.
\end{itemize}
\end{theorem}
\begin{proof}
Let ${\cal C}$ be a minor-closed graph class. If ${\cal C}$ has bounded local treewidth, the result follows from \cref{thm:FPT}.  Otherwise, \cref{thm:Eppstein} implies that ${\cal C}$ contains all apex graphs, and by \cref{thm:hard-apex} the \textsc{TSS-time}$|_{{\cal C}}$ problem is  \NP-complete for every fixed $k \geq 4$ and $ \tau^{*} = 2$.
\end{proof}

\b{Note that, since both \cref{thm:hard-apex} and \cref{thm:FPT} apply to the \textsc{GTSS-time} problem, the same dichotomy above applies to the generalized version as well. Also, since 2-\textsc{Neighbor Bootstrap Percolation-time} is a particular case of \textsc{GTSS-time}, \cref{cor:P3-hard-apex} implies the same dichotomy for the former problem, with the constraint on $\tau^*$ being irrelevant.} 
\section{Maximum TSS-time is linear-time solvable in trees}\label{sec:PolyTrees}

In this section, we obtain an $\Ocal(n)$-time algorithm and an $\Ocal(n^2)$-time algorithm for the maximization versions of \pname{TSS-time} and \pname{GTSS-time} in trees, respectively.  That is, for the problems in which the objective is to compute the maximum activation time $t_\tau(T)$ of a given tree $T$ and a (generalized) threshold function $\tau$ in $T$.

Let us begin with \pname{TSS-time}.
Given a tree $T$ and a threshold function $\tau$ in $T$, we say that a vertex $v$ is \emph{saturated} if $\tau(v)=d(v)$; otherwise, it is \emph{non-saturated}.
Clearly, a saturated vertex $v$ is activated if and only if it is in the target set or all its neighbors are activated. In other words, a saturated vertex outside the target set cannot help to activate other vertices.

Given a tree $T$ and two adjacent vertices $w$ and $x$, let $T(w,x)$ be the subtree containing $x$ obtained from $T$ by removing the edge $wx$.
Also let $T[w,x]$ be the subtree obtained from $T(w,x)$ by adding vertex $w$ and edge $wx$.

\begin{lemma}\label{lema1}
Let $T=(V,E)$ be a tree with at least two vertices, $\tau$ be a threshold function in $T$, $v$ be a leaf of $T$, and $w$ be the only neighbor of $v$.
There exists a proper subset $S\subsetneq V$ such that $v,w\not\in S$ and $I(S)=S$, and $S\cup\{v\}$ is a target set (that is $H(S\cup\{v\})=V$) which activates $w$ at time 1.
\end{lemma}

\begin{proof}
We prove the lemma by induction on the number $n$ of vertices of $T$. If $n=2$, $T$ contains exactly the two vertices $v$ and $w$ and the edge $vw$. Moreover $\tau(v)=\tau(w)=1$ (recall that $1\leq\tau(x)\leq d(x)$ for every vertex $x$ of $T$). Taking $S=\emptyset$, we are done, since $I(\emptyset)=\emptyset$ and $I(\{v\})=\{v,w\}=V(T)$.

Now, fix $n>2$, suppose that the lemma is true for every tree $T$ with less than $n$ vertices, and we will prove that the lemma is also true for trees on $n$ vertices.
Let $v$ be any leaf of $T$ and let $w$ be the only neighbor of $v$. Since $n>2$, $d(w)\geq 2$. Let $x_1,\ldots,x_{d(w)-1}$ be the neighbors of $w$ distinct from $v$.
In the following, notice that the sets $[\tau(w)-1]$ and $[\tau(w),d(w)-1]$ may be empty.

We will construct a proper subset $S$ of $V(T)$ satisfying the conditions of the lemma. Firstly let $S=\emptyset$.
If $\tau(w)\geq 2$, add to $S$ all the vertices in $T(w,x_i)$
for every $i\in[\tau(w)-1]$.
If $\tau(w)<d(w)$, fix $k\in[\tau(w),d(w)-1]$.
By the induction hypothesis, since the subtree $T[w,x_k]$ has less than $n$ vertices and $w$ is a leaf of $T[w,x_k]$, there exists a set $S_k$ such that $w,x_k\not\in S_k$, $I(S_k)=S_k$, and $H(S_k\cup\{w\})\supseteq V(T[w,x_k])$. With this, add $S_k$ to $S$ for every $k\in [\tau(w),d(w)-1]$.

By construction, we have that $v,w\not\in S$. We first prove that $I(S)=S$ in $T$. Notice that the only neighbors of $w$ in $S$ are in $\{x_i \mid i\in[\tau(w)-1]\}$, which cannot activate $w$, since its threshold is $\tau(w)$. Then, all the vertices in the subtrees $T(w,x_i)$ for all $i\in[\tau[w]-1]$ together cannot activate $w$. Moreover, $\bigcup_{k=\tau(w)}^{d(w)-1} S_{k}$ (this set may be empty) cannot activate any vertex in $\{x_k \mid k\in[\tau(w),d(w)-1]\}$, since $w,x_k\not\in S_k$ and $I(S_k)=S_k$ for $k\in[\tau(w),d(w)-1]$. Consequently, $S$ cannot activate $w$ and $I(S)=S$.

Now we prove that $H(S\cup\{v\})=V(T)$, that is, that $S\cup\{v\}$ is a target set of $T$.
Firstly notice that $v$ together with all $x_i$'s with $i\in[\tau(w)-1]$ activate $w$ at time 1, since its threshold is $\tau(w)$. Moreover, recall that $H(S_k\cup\{w\})=V(T[w,x_k])$. Therefore, all vertices in the subtrees $T[w,x_k]$ for $k \in [\tau(w),d(w)-1]$ are activated in the process and consequently $H(S\cup\{v\})=V(T)$.
\end{proof}

\cref{fig:Fig8} shows an example of the configuration considered in \cref{lema1}: a set $S$ such that $S\cup\{v\}$ is a target set, but $I(S)=S$.

\begin{figure}[ht]
\centering\scalebox{0.8}{
\begin{tikzpicture}[scale=1.2]
\tikzstyle{vertex1}=[draw,circle,fill=white,minimum size=15pt,inner sep=2pt]
\tikzstyle{vertex2}=[draw,circle,fill=black!25,minimum size=15pt,inner sep=2pt]

\node at (0.4,5) {\textcolor{red}{5}};
\node at (2.3,7.3) {\textcolor{red}{3}};
\node at (3.3,7.3) {\textcolor{red}{3}};
\node at (4.3,7.3) {\textcolor{red}{3}};
\node at (5.3,7.3) {\textcolor{red}{3}};
\node at (2.3,3.3) {\textcolor{red}{3}};
\node at (3.3,3.3) {\textcolor{red}{3}};
\node at (4.3,3.3) {\textcolor{red}{3}};
\node at (5.3,3.3) {\textcolor{red}{3}};

\node[vertex1] (v) at (-3,5) {$v$};
\node[vertex1] (w) at (0,5) {$w$};
\node[vertex2] (x1) at (-1,6) {$x_1$};
\node[vertex2] (x2) at (0,6) {$x_2$};
\node[vertex2] (y1) at (-1,7) {};
\node[vertex2] (y2) at (0,7) {};
\node[vertex2] (z1) at (-1,8) {};
\node[vertex2] (z2) at (0,8) {};
\node[vertex2] (x3) at (-1,4) {$x_3$};
\node[vertex2] (x4) at (0,4) {$x_4$};
\node[vertex2] (y3) at (-1,3) {};
\node[vertex2] (y4) at (0,3) {};
\node[vertex2] (z3) at (-1,2) {};
\node[vertex2] (z4) at (0,2) {};
\node[vertex1] (x5) at (2,7) {$x_5$};
\node[vertex1] (x6) at (2,3) {$x_6$};

\path[-,thick]
(w) edge (v) edge (x1) edge (x2) edge (x3) edge (x4) edge (x5) edge (x6)
(y1) edge (x1) edge (z1) (y2) edge (x2) edge (z2)
(y3) edge (x3) edge (z3) (y4) edge (x4) edge (z4);

\node[vertex2] (a5) at (2,8) {}; \path[-,thick] (x5) edge (a5);
\node[vertex2] (a5) at (2,6) {}; \path[-,thick] (x5) edge (a5);
\node[vertex2] (b5) at (2,5) {}; \path[-,thick] (b5) edge (a5);
\node[vertex1] (a5) at (3,7) {}; \path[-,thick] (x5) edge (a5);
\node[vertex2] (x5) at (3,8) {}; \path[-,thick] (x5) edge (a5);
\node[vertex2] (x5) at (3,6) {}; \path[-,thick] (x5) edge (a5);
\node[vertex1] (x5) at (4,7) {}; \path[-,thick] (x5) edge (a5);
\node[vertex2] (a5) at (4,8) {}; \path[-,thick] (x5) edge (a5);
\node[vertex2] (a5) at (4,6) {}; \path[-,thick] (x5) edge (a5);
\node[vertex2] (b5) at (4,5) {}; \path[-,thick] (b5) edge (a5);
\node[vertex1] (a5) at (5,7) {}; \path[-,thick] (x5) edge (a5);
\node[vertex2] (x5) at (5,8) {}; \path[-,thick] (x5) edge (a5);
\node[vertex2] (x5) at (5,6) {}; \path[-,thick] (x5) edge (a5);

\node[vertex2] (a6) at (2,2) {}; \path[-,thick] (x6) edge (a6);
\node[vertex2] (a6) at (2,4) {}; \path[-,thick] (x6) edge (a6);
\node[vertex1] (a6) at (3,3) {}; \path[-,thick] (x6) edge (a6);
\node[vertex2] (x6) at (3,2) {}; \path[-,thick] (x6) edge (a6);
\node[vertex2] (x6) at (3,4) {}; \path[-,thick] (x6) edge (a6);
\node[vertex2] (y6) at (3,5) {}; \path[-,thick] (x6) edge (y6);
\node[vertex1] (x6) at (4,3) {}; \path[-,thick] (x6) edge (a6);
\node[vertex2] (a6) at (4,2) {}; \path[-,thick] (x6) edge (a6);
\node[vertex2] (a6) at (4,4) {}; \path[-,thick] (x6) edge (a6);
\node[vertex1] (a6) at (5,3) {}; \path[-,thick] (x6) edge (a6);
\node[vertex2] (x6) at (5,2) {}; \path[-,thick] (x6) edge (a6);
\node[vertex2] (x6) at (5,4) {}; \path[-,thick] (x6) edge (a6);
\node[vertex2] (y6) at (5,5) {}; \path[-,thick] (x6) edge (y6);
\end{tikzpicture}}

\caption{An example of a tree $T$ and vertices $v$ and $w$ as in the statement of \cref{lema1}. The vertices of $S$ are shown in \b{gray}. The relevant thresholds are in red. Notice that $I(S)=S$ (i.e., $S$ does not activate any vertex), but $S\cup\{v\}$ is a target set (i.e., $H(S\cup\{v\})=V(T)$).}
%\ig{this figure is only cited in the next section (and by mistake), so we should cite it here} \ig{by the way, isn't this figure too big?}
\label{fig:Fig8}
\end{figure}
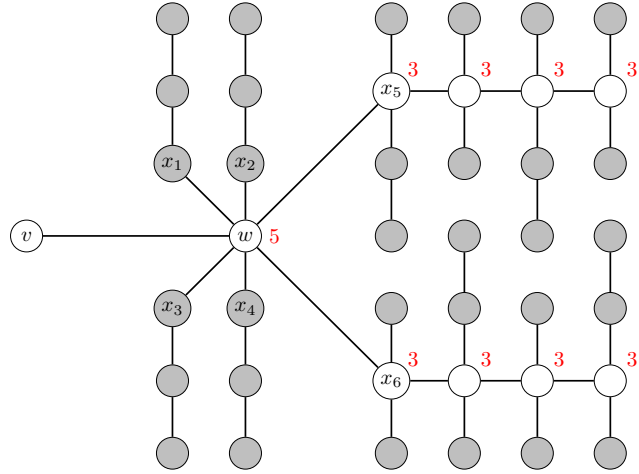 

From \cref{lema1}, we obtain the following lemma for threshold functions in trees.

\begin{lemma}\label{lema1b}
Let $T=(V,E)$ be a tree with at least two vertices and $\tau$ be a threshold function in $T$. For any path $P=(v_0,v_1,\ldots,v_p)$ with $p\geq 1$ in $T$ with $v_0$ being a leaf and all internal vertices being non-saturated, there exists a target set $S_P$ of $T$ which contains $v_0$ and activates $v_i$ at time $i$, for every $i\in[p]$.
\end{lemma}

\begin{proof}
We prove the lemma by induction on the number $p$ of edges in $P$.
If $p=1$, $P$ has only two vertices $v_0$ and $v_1$, where $v_0$ is a leaf and $v_1$ is the only neighbor of $v_0$, and we are done by \cref{lema1}.

Now fix $p\geq 2$ and suppose that the lemma is true for every path with less than $p$ edges.
Let $P=(v_0,v_1,\ldots,v_p)$ be a path with $p$ edges such that $v_0$ is a leaf and all internal vertices are non-saturated. Let us prove that the lemma is true for $P$.
Let $T'$ be the subtree containing $v_0$ obtained from $T$ by removing the edge $v_{p-1}v_p$. Since $v_{p-1}$ is non-saturated, $\tau(v_{p-1})$ is strictly smaller than the degree of $v_{p-1}$ in $T$, and consequently it is smaller than or equal to the degree of $v_{p-1}$ in $T'$.
With this, let $\tau'$ be the threshold function in $T'$ such that $\tau'(u)=\tau(u)$ for every vertex of $T'$.

Since the path $P'=(v_0,v_1,\ldots,v_{p-1})$ in $T'$ (with threshold function $\tau'$) has less than $p$ edges, we have by induction that there exists a target set $S'$ of $T'$ which contains $v_0$ and activates $v_i$ at time $i$, for every $i\in[p-1]$.

Now let $T''=T[v_{p-1},v_p]$ and let $\tau''$ be the threshold function in $T''$ such that $\tau''(v_{p-1})=1$ and $\tau''(u)=\tau(u)$ for every vertex $u\in V(T'')\setminus\{v_{p-1}\}$. Since $v_{p-1}$ is a leaf of $T''$ and $v_p$ is the only neighbor of $v_{p-1}$ in $T''$, we can apply \cref{lema1} and obtain a vertex subset $S''$ in $T''$ such that $v_{p-1},v_p\not\in S''$ and $S''$ does not activate vertices in $T''$, and such that $S''\cup\{v_{p-1}\}$ is a target set of $T''$ which activates $v_p$ at time 1.

With this, let $S_P=S'\cup S''$. By construction, we have that $S_P$ contains $v_0$ and activates all vertices in $T'$, since it contains $S'$, activating $v_i$ at time $i$ for every $i\in[p-1]$. Finally, since $S_P$ contains $S''$ and activates $v_{p-1}$ at time $p-1$, we have that $S_P$ also activates all vertices in $T''$ (and consequently $S_P$ is a target set of $T$) and activates $v_p$ at time $p$.
\end{proof}

Given a tree $T$ and a threshold function $\tau$ in $T$, let $F_{T,\tau}$ be the forest obtained from $T$ in the following way: first remove all saturated vertices, and then, for every saturated vertex $v$ in $T$ and every non-saturated neighbor $w$ of $v$ in $T$, create a new vertex $v_w$ and add the edge $v_ww$ to $F_{T,\tau}$.

\begin{theorem}\label{teo-trees1}
For any tree $T$ and threshold function $\tau$ in $T$, the maximum activation time $t_\tau(T)$ is the maximum diameter among the trees in the forest $F_{T,\tau}$.
Consequently, \pname{TSS-time} is linear-time solvable in trees.
\end{theorem}

\begin{proof}
Consider a target set $S$ of a tree $T$ which activates a vertex $v$ at time $t$. Then there exists a path $P=(v_0,v_1,\ldots,v_{t-1},v_t=v)$ in $T$ of vertices activated by $S$ at times $0,1,\ldots,t-1,t$, respectively.
Since a saturated vertex is activated if it is in the target set or if all its neighbors are activated, all internal vertices in the path $P$ are non-saturated.

Now consider a path $P=(v_0,v_1,\ldots,v_t)$ such that all its internal vertices are non-saturated. Let $T'=T[v_0,v_1]$. Since $v_0$ is a leaf of $T'$, by \cref{lema1b} there exists a target set $S'$ of $T'$ which contains $v_0$ and activates $v_i$ at time $i$ for every $i\in[t]$. Let $S$ be the set obtained from $S'$ by adding all vertices in $V(T) \setminus V(T')$. Therefore, $S$ is a target set of $T$ which contains $v_0$ and activates $v_i$ at time $i$ for every $i\in[t]$.

Thus, $T$ has maximum activation time at least $t$ if and only if there exists a path $P$ with $t$ edges in $T$ such that all its internal vertices are non-saturated.
%Moreover, the existence of such a path $P$ with $t$ edges implies that the maximum activation time is at least $t$.
Then, \b{by construction of $F_{T,\tau}$}, the maximum activation time $t_\tau(T)$ is equal to the maximum diameter among the trees in the forest $F_{T,\tau}$.
Since the diameter of a tree can be computed in linear time, \b{and the forest $F_{T,\tau}$ can be clearly constructed in linear time}, we have that \pname{TSS-time} is linear-time solvable in trees.
\end{proof}

\medskip

Let us now focus on the \pname{GTSS-time} problem.
Given a graph $G$, a generalized threshold function $\tau$ in $G$, a vertex subset $S_0\subseteq V(G)$, and a vertex $v$ of $G$, recall that $t_{\tau}(v,S_0)$ is the minimum integer $k$ such that $v\in I^k(S_0)$, or $t_{\tau}(v,S_0)=\infty$ if $v\not\in H(S_0)$. By applying the algorithm \pname{Activation-Times} with input set $S_0$, we have that $H(S_0)$ and $t_{\tau}(v,S_0)$ for every vertex $v$ can be computed in time $\Ocal(m+n)$.

We first prove the auxiliary lemma below.
\b{Let in this section} $V_{\sf f}$ be the set of \emph{forced} vertices by the threshold function, that is, the set of vertices $u$ of $T$ with $\tau(u)>d(u)$.
\begin{lemma}\label{lema2b}
Let $T$ be a tree and $\tau$ be a generalized threshold function in $T$.
Let $S_0\supseteq V_{\sf f}$.
For every vertex $v\in H(S_0)$, there exists a target set $S_v\supseteq S_0$ such that $t_\tau(v,S_v)=t_\tau(v,S_0)$.
\end{lemma}

\begin{proof}
Initially let $S_v=S_0$. If $S_v$ is a target set, we are done.
Otherwise, \b{we iteratively apply the following procedure}:

($\maltese$) Let $T_v$ be the maximal subtree of $T$ containing $v$ and all vertices in $H(S_v)$. Let $u$ be any vertex of $T_v$ with a neighbor $w$ outside $T_v$. Let $T_w$ be the maximal subtree of $T[u,w]$ containing $u$ and $w$ with no vertex in $H(S_v)$ other than $u$. Also let $\tau_w$ be such that $\tau_w(u)=1$, $\tau_w(w)=\tau(w)-|N(w)\cap H(S_v)|+1$, and $\tau_w(x)=\tau(x)-|N(x)\cap H(S_v)|$ for every vertex $x$ of $T_w$, except $u$ and $w$. Notice that $\tau_w$ is a threshold function of $T_w$ and $u$ is a leaf of $T_w$. Then, applying \cref{lema1}, we have that there exists a set $S_w$ in $T_w$ that activates no vertex in $T$ and such that $S_w\cup\{u\}$ is a target set of $T_w$. Add $S_w$ to $S_v$. Notice that $t_\tau(v,S_v)=t_\tau(v,S_0)$.

Repeating ($\maltese$) until $H(S_v)=V(T)$, we obtain a target set $S_v$ such that $t_\tau(v,S_v)=t_\tau(v,S_0)$, and the lemma follows.
\end{proof}

We now explain how to compute $t_{\tau}(T)$ in time $\Ocal(n^2)$ \b{for a given pair $(T,\tau)$}.
For this, we define, for every vertex $v$ of $T$, $t_\tau(v)$ as the maximum $t_\tau(v,S_0)$ among all target sets $S_0$ of $T$.
Start by computing $H(V_{\sf f})$ and compute the time $t_\tau(v)=t_{\tau}(v,V_{\sf f})$ for every $v\in H(V_{\sf f})$. We define the \emph{beginning time} $b(v)$ of every vertex $v \in V(T) \setminus V_{\sf f}$ as the maximum $t_\tau(w)$ among all neighbors $w$ of $v$ in $H(V_{\sf f})$, if there is one.
Otherwise, let $b(v)=0$.

As before, we say that a vertex $v$ is \emph{saturated} if $\tau(v)=d(v)$; otherwise it is \emph{non-saturated}.
As before, a saturated vertex is activated if and only if it is in the target set or if all its neighbors are activated. In other words, a saturated vertex outside the target set cannot help to activate other vertices.

A \emph{non-saturated path} in the tree $T$ is a path such that all its vertices are non-saturated (including the endpoints) and are outside $H(V_{\sf f})$.
Here we allow paths with only one vertex (and no edge) from a vertex $v$ to itself.
%In the following, we prove that, for every non-saturated vertex $v$ outside $H(V_{\sf f})$, the maximum activation time $t_{\tau}(v)$ is equal to the maximum value, among all non-saturated paths $P$ with an endpoint in $v$, of $b(u)+|P|$, where $|P|$ is the number of vertices in the path $P$ and $u$ is the other endpoint of $P$. Moreover, for every saturated vertex $w$ outside $H(V_{\sf f})$, we prove that $t_{\tau}(w)=1+\max\{t_{\tau}(v)\}$ among all non-saturated neighbors $v$ of $w$, if there is one; otherwise $t_{\tau}(w)=1$.

\begin{lemma}\label{lem:tau-generalized}
Let $T$ be a tree, $\tau$ be a generalized threshold function in $T$, and $v\in V(T)$. Let $V_{\sf f}$ be the set of vertices $u$ of $T$ with $\tau(u)>d(u)$. If $v\in H(V_{\sf f})$, then $t_{\tau}(v)=t_{\tau}(v,V_{\sf f})$.
If $v$ is non-saturated outside $H(V_{\sf f})$, then $t_{\tau}(v)=\max\{|P|+b(u) : P\mbox{ is a non-saturated path with an endpoint in $v$, where $u$ is the other endpoint}\}$.
If $v$ is saturated outside $H(V_{\sf f})$, then $t_{\tau}(v)=1+\max\{t_{\tau}(u)\}$ among all non-saturated neighbors $u$ of $v$, if there is one; otherwise $t_{\tau}(v)=1$.
\end{lemma}

\begin{proof}
Let $v$ be a vertex of $T$. Suppose first that $v\in H(V_{\sf f})$.
From \cref{lema2b}, there exists a target set $S_v$ such that $t_\tau(v,S_v)=t_\tau(v,V_{\sf f})$. Therefore, $t_{\tau}(v)\geq t_{\tau}(v,V_{\sf f})$. Moreover, since the vertices of $V_{\sf f}$ must be in the target set $S'_v$ with $t_{\tau}(v,S'_v)=t_{\tau}(v)$, then $t_{\tau}(v)= t_{\tau}(v,S'_v)\leq t_{\tau}(v,V_{\sf f})$, and we are done.

Now suppose that $v\not\in H(V_{\sf f})$ and $v$ is saturated.
If all neighbors of $v$ are saturated, then $t_{\tau}(v)=1$, since all its neighbors must be in the target set (otherwise $v$ must be in the target set) and $V(T)\setminus\{v\}$ is a target set activating $v$ at time 1.
Moreover, if $v$ is saturated and has at least one non-saturated neighbor, then $t_{\tau}(v)=1+\max\{t_{\tau}(u) : \mbox{ $u$ is a non-saturated neighbor of $v$}\}$, since $v$ cannot be activated before its neighbors (unless it is in the target set).

Finally, consider a non-saturated vertex $v$. Consider a target set $S_0$ which activates $v$ at time $t$. We want to show that there exists a path $P$ of non-saturated vertices such that $t=b(u)+|P|$, where $u$ and $v$ are the endpoints of $P$. First notice that there exists a path $P'=(u_0,u_1,u_2,\ldots,u_{t-1},v)$ in the tree $T$ whose vertices are activated by $S_0$ at times $0,1,2,\ldots,t-1,t$, respectively.
Since vertices with $\tau(v)\geq d(v)$ cannot help to activate other vertices at time greater than 1, all vertices in the path $P'$, except $u_0$, are non-saturated.
We may assume that there exists $0\leq k<t$ such that $u_1,\ldots,u_k\in H(V_{\sf f})$ and $u_{k+1},\ldots,u_{t-1}\not\in H(V_{\sf f})$. This is because every vertex of $H(V_{\sf f})$ with activation time $k+1$ was activated by a vertex of $H(V_{\sf f})$ with activation time $k$.
Therefore $b(u_{k+1})\geq t_{\tau}(u_k)=k$. The subpath $P=(u_{k+1},\ldots,u_{t-1},v)$ of $P'$ is a non-saturated path and has size $|P|=t-k$. Then the activation time of $v$ in the process of the target set $S_0$ is equal to $t=k+(t-k)=b(u)+|P|$, where $u=u_{k+1}$ is the endpoint of $P$ distinct from $v$.

Now consider a non-saturated path $P=(u_1,\ldots,u_{\ell-1},v)$ with size $|P|=\ell$.
Recall that, by definition, $P$ has no vertex in $H(V_{\sf f})$.
We want to show that there exists a target set $S_0$ which activates $v$ at time $b(u_1)+\ell$. Initially let $S_0=V_{\sf f}$.
%If $u_1$ has a neighbor in $H(V_{\sf f})$, then it has a neighbor $w\in H(V_{\sf f})$ with $t_\tau(w,V_{\sf f})=b(u_1)$.
Since $u_1$ is non-saturated, we can add to $S_0$ $\tau(u_1)-|N(u_1)\cap H(S_0)|$ neighbors of $u_1$ outside $H(S_0)$ distinct from $u_2$. With this, $S_0$ activates $u_1$ at time $b(u_1)+1$. Again, since $u_2$ is non-saturated, we can add to $S_0$ $\tau(u_2)-|N(u_2)\cap H(S_0)|$ neighbors of $u_2$ outside $H(S_0)$ distinct from $u_3$. With this, $S_0$ activates $u_2$ at time $b(u_1)+2$. Following these arguments, we obtain a set $S_0$ which activates $v$ at time $b(u_1)+\ell$. From \cref{lema2b}, there exists a target set $S_v\supset S_0$ which activates $u_\ell$ at time $b(u_1)+\ell$, and we are done.
\end{proof}

\begin{figure}[ht]
\centering
\begin{tikzpicture}[scale=1.15]
\tikzstyle{vertex1}=[draw,circle,fill=blue!30,minimum size=15pt,inner sep=2pt]
\tikzstyle{vertex2}=[draw,circle,fill=black!30,minimum size=15pt,inner sep=2pt]
\tikzstyle{vertex3}=[draw,circle,fill=green!30,minimum size=15pt,inner sep=2pt]
\tikzstyle{vertex4}=[draw,circle,fill=gray!30,minimum size=15pt,inner sep=2pt]
\tikzstyle{vertex5}=[draw,circle,fill=white!50,minimum size=15pt,inner sep=2pt]
\tikzstyle{vertex6}=[draw,circle,fill=orange!50,minimum size=15pt,inner sep=2pt]

\node at (0,3.4)   {\textcolor{red}{2}};
\node at (1,3.4)   {\textcolor{red}{2}};
\node at (2,3.4)   {\textcolor{red}{1}};
\node at (3,3.4)   {\textcolor{red}{0}};
\node at (4,3.4)   {\textcolor{red}{0}};
\node at (5,3.4)   {\textcolor{red}{2}};
\node at (6,3.4)   {\textcolor{red}{2}};
\node at (7,3.4)   {\textcolor{red}{2}};
\node at (8,3.4)   {\textcolor{red}{0}};
\node at (9,3.4)   {\textcolor{red}{2}};
\node at (10,3.4)  {\textcolor{red}{1}};
\node at (10.3,2.3){\textcolor{red}{0}};
\node at (0.3,2.3) {\textcolor{red}{2}};
\node at (1.3,2.3) {\textcolor{red}{2}};
\node at (2.3,2.3) {\textcolor{red}{2}};
\node at (3.3,2.3) {\textcolor{red}{2}};
\node at (4.3,2.3) {\textcolor{red}{2}};
\node at (5.3,2.3) {\textcolor{red}{2}};
\node at (6.3,2.3) {\textcolor{red}{2}};
\node at (7.3,2.3) {\textcolor{red}{2}};
\node at (8.0,2.4) {\textcolor{red}{2}};
\node at (9.3,2.3) {\textcolor{red}{2}};
\node at (0.3,1.3) {\textcolor{red}{2}};
\node at (1.3,1.3) {\textcolor{red}{2}};
\node at (2.3,1.3) {\textcolor{red}{2}};
\node at (3.3,1.3) {\textcolor{red}{2}};
\node at (4.3,1.3) {\textcolor{red}{2}};
\node at (5.3,1.3) {\textcolor{red}{2}};
\node at (6.3,1.3) {\textcolor{red}{2}};
\node at (7.3,1.3) {\textcolor{red}{2}};
\node at (8.3,1.3) {\textcolor{red}{2}};
\node at (9.3,1.3) {\textcolor{red}{2}};
\node at (10.3,1.3){\textcolor{red}{1}};
\node at (0.0,-.4) {\textcolor{red}{1}};
\node at (1.0,-.4) {\textcolor{red}{2}};
\node at (2.0,-.4) {\textcolor{red}{2}};
\node at (3.0,-.4) {\textcolor{red}{0}};
\node at (4.0,-.4) {\textcolor{red}{0}};
\node at (5.0,-.4) {\textcolor{red}{0}};
\node at (6.0,-.4) {\textcolor{red}{2}};
\node at (7.0,-.4) {\textcolor{red}{2}};
\node at (8.0,-.4) {\textcolor{red}{1}};
\node at (9.0,-.4) {\textcolor{red}{3}};
\node at (10 ,-.4) {\textcolor{red}{2}};

\node[vertex2] (a) at (8,1) {$0$};
%\node[vertex2] (c) at (8.4,0) {$0$}; \path[-,thick] (c) edge (a);
%\node[vertex2] (c) at (7.7,0) {$0$}; \path[-,thick] (c) edge (a);
\node[vertex1] (b) at (7,1) {$1$}; \path[-,thick] (a) edge (b);
\node[vertex2] (c) at (7,0) {$0$}; \path[-,thick] (c) edge (b);
\node[vertex1] (a) at (6,1) {$2$}; \path[-,thick] (a) edge (b);
\node[vertex2] (c) at (6,0) {$0$}; \path[-,thick] (c) edge (a);
\node[vertex1] (b) at (5,1) {$3$}; \path[-,thick] (a) edge (b);
\node[vertex1] (c) at (5,0) {$1$}; \path[-,thick] (c) edge (b);
\node[vertex1] (a) at (4,1) {$4$}; \path[-,thick] (a) edge (b);
\node[vertex1] (c) at (4,0) {$1$}; \path[-,thick] (c) edge (a);
\node[vertex1] (b) at (3,1) {$5$}; \path[-,thick] (a) edge (b);
\node[vertex1] (c) at (3,0) {$1$}; \path[-,thick] (c) edge (b);
\node[vertex1] (a) at (2,1) {$6$}; \path[-,thick] (a) edge (b);
\node[vertex2] (c) at (2,0) {$0$}; \path[-,thick] (c) edge (a);
\node[vertex1] (b) at (1,1) {$7$}; \path[-,thick] (a) edge (b);
\node[vertex2] (c) at (1,0) {$0$}; \path[-,thick] (c) edge (b);
\node[vertex3] (a) at (1,2) {$8$}; \path[-,thick] (a) edge (b);
\node[vertex4] (d) at (1,3) {$0$}; \path[-,thick] (d) edge (a);
\node[vertex3] (b) at (2,2) {$9$}; \path[-,thick] (a) edge (b);
\node[vertex4] (c) at (2,3) {$0$}; \path[-,thick] (c) edge (b);
\node[vertex3] (a) at (3,2) {$10$}; \path[-,thick] (a) edge (b);
\node[vertex1] (c) at (3,3) {$1$}; \path[-,thick] (c) edge (a);
\node[vertex3] (b) at (4,2) {$11$}; \path[-,thick] (a) edge (b);
\node[vertex1] (c) at (4,3) {$1$}; \path[-,thick] (c) edge (b);
\node[vertex3] (a) at (5,2) {$12$}; \path[-,thick] (a) edge (b);
\node[vertex2] (c) at (5,3) {$0$}; \path[-,thick] (c) edge (a);
\node[vertex3] (b) at (6,2) {$13$}; \path[-,thick] (a) edge (b);
\node[vertex2] (c) at (6,3) {$0$}; \path[-,thick] (c) edge (b);
\node[vertex3] (a) at (7,2) {$14$}; \path[-,thick] (a) edge (b);
\node[vertex2] (c) at (7,3) {$0$}; \path[-,thick] (c) edge (a);
\node[vertex3] (b) at (8,2) {$15$}; \path[-,thick] (a) edge (b);
\node[vertex1] (e) at (9,3) {$3$}; \path[-,thick] (e) edge (b);
\node[vertex6] (a) at (9,2) {$16$}; \path[-,thick] (a) edge (b);
\node[vertex4] (c) at (9,1) {$0$}; \path[-,thick] (c) edge (a);

\node[vertex5] (a) at (9,0) {$1$}; \path[-,thick] (c) edge (a);
\node[vertex4] (c) at (8,0) {$0$}; \path[-,thick] (c) edge (a);
\node[vertex4] (c) at (10,0){$0$}; \path[-,thick] (c) edge (a);
\node[vertex5] (a) at (10,1){$1$}; \path[-,thick] (c) edge (a);

\node[vertex5] (a) at (0,3) {$1$}; \path[-,thick] (d) edge (a);
\node[vertex4] (d) at (0,2) {$0$}; \path[-,thick] (d) edge (a);
\node[vertex5] (a) at (0,1) {$1$}; \path[-,thick] (d) edge (a);
\node[vertex4] (d) at (0,0) {$0$}; \path[-,thick] (d) edge (a);

\node[vertex1] (a) at (8,3) {$1$}; \path[-,thick] (e) edge (a);
\node[vertex1] (a) at (10,3){$2$}; \path[-,thick] (e) edge (a);
\node[vertex1] (e) at (10,2){$1$}; \path[-,thick] (e) edge (a);

\end{tikzpicture}

\caption{A tree $G$ with maximum activation time 16. The thresholds are in red, $V_{\sf f}$ is in \b{dark gray}, and $H(V_{\sf f})\setminus S_0$ is in blue. A maximum path of non-saturated vertices is in  green. The vertex with maximum activation time 16 is in orange (notice that it is saturated). The numbers inside the vertices are their activation times. The target set $S_0$ with maximum activation time is in \b{dark gray and light gray} (vertices with activation time 0).}
\label{fig:Fig9}
\end{figure} 

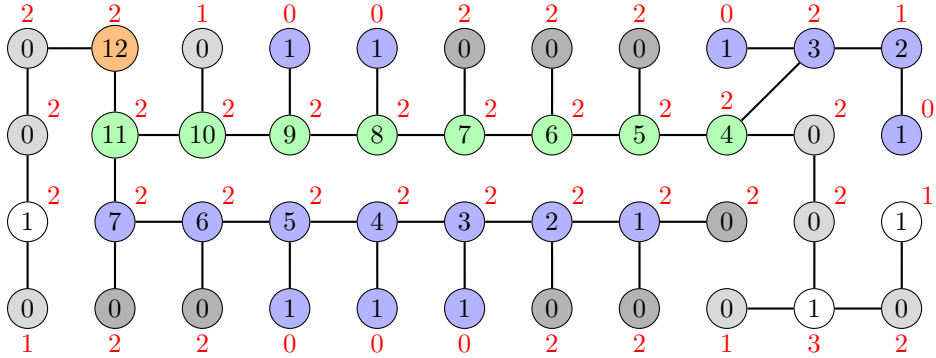
\begin{figure}[ht]
\centering
\begin{tikzpicture}[scale=1.15]
\tikzstyle{vertex1}=[draw,circle,fill=blue!30,minimum size=15pt,inner sep=2pt]
\tikzstyle{vertex2}=[draw,circle,fill=black!30,minimum size=15pt,inner sep=2pt]
\tikzstyle{vertex3}=[draw,circle,fill=green!30,minimum size=15pt,inner sep=2pt]
\tikzstyle{vertex4}=[draw,circle,fill=gray!30,minimum size=15pt,inner sep=2pt]
\tikzstyle{vertex5}=[draw,circle,fill=white!50,minimum size=15pt,inner sep=2pt]
\tikzstyle{vertex6}=[draw,circle,fill=orange!50,minimum size=15pt,inner sep=2pt]

\node at (0,3.4)   {\textcolor{red}{2}};
\node at (1,3.4)   {\textcolor{red}{2}};
\node at (2,3.4)   {\textcolor{red}{1}};
\node at (3,3.4)   {\textcolor{red}{0}};
\node at (4,3.4)   {\textcolor{red}{0}};
\node at (5,3.4)   {\textcolor{red}{2}};
\node at (6,3.4)   {\textcolor{red}{2}};
\node at (7,3.4)   {\textcolor{red}{2}};
\node at (8,3.4)   {\textcolor{red}{0}};
\node at (9,3.4)   {\textcolor{red}{2}};
\node at (10,3.4)  {\textcolor{red}{1}};
\node at (10.3,2.3){\textcolor{red}{0}};
\node at (0.3,2.3) {\textcolor{red}{2}};
\node at (1.3,2.3) {\textcolor{red}{2}};
\node at (2.3,2.3) {\textcolor{red}{2}};
\node at (3.3,2.3) {\textcolor{red}{2}};
\node at (4.3,2.3) {\textcolor{red}{2}};
\node at (5.3,2.3) {\textcolor{red}{2}};
\node at (6.3,2.3) {\textcolor{red}{2}};
\node at (7.3,2.3) {\textcolor{red}{2}};
\node at (8.0,2.4) {\textcolor{red}{2}};
\node at (9.3,2.3) {\textcolor{red}{2}};
\node at (0.3,1.3) {\textcolor{red}{2}};
\node at (1.3,1.3) {\textcolor{red}{2}};
\node at (2.3,1.3) {\textcolor{red}{2}};
\node at (3.3,1.3) {\textcolor{red}{2}};
\node at (4.3,1.3) {\textcolor{red}{2}};
\node at (5.3,1.3) {\textcolor{red}{2}};
\node at (6.3,1.3) {\textcolor{red}{2}};
\node at (7.3,1.3) {\textcolor{red}{2}};
\node at (8.3,1.3) {\textcolor{red}{2}};
\node at (9.3,1.3) {\textcolor{red}{2}};
\node at (10.3,1.3){\textcolor{red}{1}};
\node at (0.0,-.4) {\textcolor{red}{1}};
\node at (1.0,-.4) {\textcolor{red}{2}};
\node at (2.0,-.4) {\textcolor{red}{2}};
\node at (3.0,-.4) {\textcolor{red}{0}};
\node at (4.0,-.4) {\textcolor{red}{0}};
\node at (5.0,-.4) {\textcolor{red}{0}};
\node at (6.0,-.4) {\textcolor{red}{2}};
\node at (7.0,-.4) {\textcolor{red}{2}};
%\node at (7.7,-.4) {\textcolor{red}{2}};
%\node at (8.4,-.4) {\textcolor{red}{2}};
\node at (8.0,-.4) {\textcolor{red}{1}};
\node at (9.0,-.4) {\textcolor{red}{3}};
\node at (10 ,-.4) {\textcolor{red}{2}};

\node[vertex2] (a) at (8,1) {$0$};
%\node[vertex2] (c) at (8.4,0) {$0$}; \path[-,thick] (c) edge (a);
%\node[vertex2] (c) at (7.7,0) {$0$}; \path[-,thick] (c) edge (a);
\node[vertex1] (b) at (7,1) {$1$}; \path[-,thick] (a) edge (b);
\node[vertex2] (c) at (7,0) {$0$}; \path[-,thick] (c) edge (b);
\node[vertex1] (a) at (6,1) {$2$}; \path[-,thick] (a) edge (b);
\node[vertex2] (c) at (6,0) {$0$}; \path[-,thick] (c) edge (a);
\node[vertex1] (b) at (5,1) {$3$}; \path[-,thick] (a) edge (b);
\node[vertex1] (c) at (5,0) {$1$}; \path[-,thick] (c) edge (b);
\node[vertex1] (a) at (4,1) {$4$}; \path[-,thick] (a) edge (b);
\node[vertex1] (c) at (4,0) {$1$}; \path[-,thick] (c) edge (a);
\node[vertex1] (b) at (3,1) {$5$}; \path[-,thick] (a) edge (b);
\node[vertex1] (c) at (3,0) {$1$}; \path[-,thick] (c) edge (b);
\node[vertex1] (a) at (2,1) {$6$}; \path[-,thick] (a) edge (b);
\node[vertex2] (c) at (2,0) {$0$}; \path[-,thick] (c) edge (a);
\node[vertex1] (b) at (1,1) {$7$}; \path[-,thick] (a) edge (b);
\node[vertex2] (c) at (1,0) {$0$}; \path[-,thick] (c) edge (b);
\node[vertex3] (a) at (1,2) {$11$}; \path[-,thick] (a) edge (b);
\node[vertex6] (d) at (1,3) {$12$}; \path[-,thick] (d) edge (a);
\node[vertex3] (b) at (2,2) {$10$}; \path[-,thick] (a) edge (b);
\node[vertex4] (c) at (2,3) {$0$}; \path[-,thick] (c) edge (b);
\node[vertex3] (a) at (3,2) {$9$}; \path[-,thick] (a) edge (b);
\node[vertex1] (c) at (3,3) {$1$}; \path[-,thick] (c) edge (a);
\node[vertex3] (b) at (4,2) {$8$}; \path[-,thick] (a) edge (b);
\node[vertex1] (c) at (4,3) {$1$}; \path[-,thick] (c) edge (b);
\node[vertex3] (a) at (5,2) {$7$}; \path[-,thick] (a) edge (b);
\node[vertex2] (c) at (5,3) {$0$}; \path[-,thick] (c) edge (a);
\node[vertex3] (b) at (6,2) {$6$}; \path[-,thick] (a) edge (b);
\node[vertex2] (c) at (6,3) {$0$}; \path[-,thick] (c) edge (b);
\node[vertex3] (a) at (7,2) {$5$}; \path[-,thick] (a) edge (b);
\node[vertex2] (c) at (7,3) {$0$}; \path[-,thick] (c) edge (a);
\node[vertex3] (b) at (8,2) {$4$}; \path[-,thick] (a) edge (b);
\node[vertex1] (e) at (9,3) {$3$}; \path[-,thick] (e) edge (b);
\node[vertex4] (a) at (9,2) {$0$}; \path[-,thick] (a) edge (b);
\node[vertex4] (c) at (9,1) {$0$}; \path[-,thick] (c) edge (a);

\node[vertex5] (a) at (9,0) {$1$}; \path[-,thick] (c) edge (a);
\node[vertex4] (c) at (8,0) {$0$}; \path[-,thick] (c) edge (a);
\node[vertex4] (c) at (10,0){$0$}; \path[-,thick] (c) edge (a);
\node[vertex5] (a) at (10,1){$1$}; \path[-,thick] (c) edge (a);

\node[vertex4] (a) at (0,3) {$0$}; \path[-,thick] (d) edge (a);
\node[vertex4] (d) at (0,2) {$0$}; \path[-,thick] (d) edge (a);
\node[vertex5] (a) at (0,1) {$1$}; \path[-,thick] (d) edge (a);
\node[vertex4] (d) at (0,0) {$0$}; \path[-,thick] (d) edge (a);

\node[vertex1] (a) at (8,3) {$1$}; \path[-,thick] (e) edge (a);
\node[vertex1] (a) at (10,3){$2$}; \path[-,thick] (e) edge (a);
\node[vertex1] (e) at (10,2){$1$}; \path[-,thick] (e) edge (a);

\end{tikzpicture}

\caption{Another target set of the same tree of \cref{fig:Fig9} (with time 12). The thresholds are in red, $V_{\sf f}$ is in dark gray, and $H(V_{\sf f})\setminus S_0$ is in blue. A maximum path of non-saturated vertices is in green. The vertex with time 12 is in orange (notice that it is saturated). The numbers inside the vertices are their activation times. The target set $S_0$ is in \b{dark gray and light gray} (vertices with time 0).}
\label{fig:Fig10}
\end{figure}

\cref{fig:Fig9} and \cref{fig:Fig10} show an example for the same tree $T$. In both figures, the \b{dark gray} vertices are the vertices of $V_{\sf f}$ (that is, vertices $v$ with $\tau(v)>d(v)$) and the blue vertices are the vertices in $H(V_{\sf f})\setminus V_{\sf f}$. Let $u$ and $v$ be the vertices with labels 8 and 15 in \cref{fig:Fig9}, respectively. Notice that all 8 vertices in the path between $u$ and $v$ (green in both figures) are non-saturated.  In both figures, The numbers inside the dark gray or blue vertices are the values of $t_{\tau}(w)$ of the vertices in $H(V_{\sf f})$. In this example, we have that $b(u)=7$ and $b(v)=3$. Moreover, $t_{\tau}(u)=3+8=11$ and $t_{\tau}(v)=7+8=15$.
The maximum times 16 and 12 in \cref{fig:Fig9} and \cref{fig:Fig10}, respectively, are achieved at saturated vertices. The maximum time $t_{\tau}(T)$ is 16, obtained by the target set of \cref{fig:Fig9}.  \cref{fig:Fig11} shows an example where the maximum time is achieved at a vertex of $H(V_{\sf f})$.

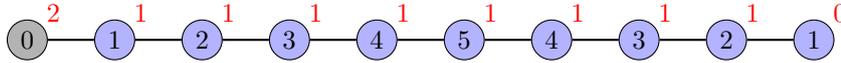
\begin{figure}[ht]
\centering
\begin{tikzpicture}[scale=1.15]
\tikzstyle{vertex1}=[draw,circle,fill=blue!30,minimum size=15pt,inner sep=2pt]
\tikzstyle{vertex2}=[draw,circle,fill=black!30,minimum size=15pt,inner sep=2pt]

\node at (0.3,1.3) {\textcolor{red}{2}};
\node at (1.3,1.3) {\textcolor{red}{1}};
\node at (2.3,1.3) {\textcolor{red}{1}};
\node at (3.3,1.3) {\textcolor{red}{1}};
\node at (4.3,1.3) {\textcolor{red}{1}};
\node at (5.3,1.3) {\textcolor{red}{1}};
\node at (6.3,1.3) {\textcolor{red}{1}};
\node at (7.3,1.3) {\textcolor{red}{1}};
\node at (8.3,1.3) {\textcolor{red}{1}};
\node at (9.3,1.3) {\textcolor{red}{0}};

\node[vertex2] (a) at (0,1) {$0$};
\node[vertex1] (b) at (1,1) {$1$}; \path[-,thick] (a) edge (b);
\node[vertex1] (a) at (2,1) {$2$}; \path[-,thick] (a) edge (b);
\node[vertex1] (b) at (3,1) {$3$}; \path[-,thick] (a) edge (b);
\node[vertex1] (a) at (4,1) {$4$}; \path[-,thick] (a) edge (b);
\node[vertex1] (b) at (5,1) {$5$}; \path[-,thick] (a) edge (b);
\node[vertex1] (a) at (6,1) {$4$}; \path[-,thick] (a) edge (b);
\node[vertex1] (b) at (7,1) {$3$}; \path[-,thick] (a) edge (b);
\node[vertex1] (a) at (8,1) {$2$}; \path[-,thick] (a) edge (b);
\node[vertex1] (b) at (9,1) {$1$}; \path[-,thick] (a) edge (b);
\end{tikzpicture}

\caption{Example with maximum activation time $t(G)=5$. All vertices belong to $H(B_0)$. The target set is $B_0$ (only the vertex in dark gray). The thresholds are in red. The numbers inside the vertices are their activation times.
%\ig{this figure is not cited in the text}}
}\label{fig:Fig11}
\end{figure}

\begin{theorem}\label{teo-trees2}
Let $T$ be a tree and $\tau$ be a generalized threshold function in $T$.
Then, $t_{\tau}(T)=\max\{t_{\tau}(v) \mid v\in V(T)\}$.
Consequently, \pname{GTSS-time} is $\Ocal(n^2)$-time solvable \b{in trees}.
\end{theorem}

\begin{proof}
Clearly $t_{\tau}(T)=\max\{t_{\tau}(v) \mid v\in V(T)\}$, since the maximum activation time must be achieved at some vertex.
In order to compute $t_\tau(T)$, we have to compute $H(V_{\sf f})$ and $b(v)$ for every vertex $v$ of $T$, which can be done in $\Ocal(n)$-time by the algorithm \pname{Activation-Times}.
With this, we have computed $t_\tau(v)$ for every vertex $v \in H(V_{\sf f})$.
Let $v$ be a non-saturated vertex outside $H(V_{\sf f})$. We can now compute a maximum non-saturated path $P$ with an endpoint in $v$ in $\Ocal(n)$-time, by a breadth-first search over non-saturated vertices outside $H(V_{\sf f})$. Thus, we can compute $t_\tau(v)$ for every non-saturated vertex outside $H(V_{\sf f})$ in time $\Ocal(n^2)$.
For saturated vertices $v$, we can compute $t_\tau(v)$ by searching locally within its neighborhood.
\end{proof}

One interesting observation is that, in \cref{teo-trees1}, the threshold values are not important, but only whether a vertex is saturated or not.
However, in \cref{teo-trees2}, the threshold values are important, since the beginning set $H(V_{\sf f})$ depends on these values.

\section{Further research}\label{sec:concl}

%\ig{Here I just list some open questions that I think that are important, and that we should try (at least a bit) to solve before submitting the paper to make it stronger. Once we will know which of these questions we really leave as open, I can write a ``proper'' conclusions section}

We introduced the \pname{Target Set Selection-Time (TSS-time)} problem and studied its computational complexity, as well as for its generalized version  (\pname{GTSS-time}), obtaining both positive and negative results. A number of interesting questions remain open. In particular, is the value of $k$ in our \NP-hardness results tight? Namely, $k=4$ in \cref{thm:NPC} and \cref{thm:hard-apex}, and $k=5$ in \cref{thm:NPC2}. For the $2$-\textsc{Neighbor Bootstrap Percolation-time} problem, non-trivial arguments were needed in order to establish such dichotomies~\cite{wg2014-tcs}, which do not seem to be easily generalizable to our problem. 

Our main result (\cref{thm:dichotomy}) is a complexity dichotomy for the \textsc{TSS-time} problem in minor-closed graph classes, as well as for its  generalized version. Within minor-closed graph classes of bounded local treewidth (for which know that the \textsc{TSS-time} problem is \FPT with parameters $k$ and $\tau^*)$, it would be very interesting to obtain an additional dichotomy  distinguishing between the polynomial-time solvable cases (such as trees, cf. \cref{teo-trees1}) and the \NP-complete ones (such as planar graphs, cf. \cref{corol-W1}). Another natural research direction is to obtain a complexity dichotomy including also graph classes that are not minor-closed. In the proof of our dichotomy (\cref{thm:dichotomy}), we crucially use \cref{thm:Eppstein}, which only applies to minor-closed graph classes.

As an ingredient in our complexity dichotomy, we proved in \cref{thm:FPT} that, if ${\cal C}$ is a graph class of bounded local treewidth, then the \textsc{GTSS-time} problem restricted to input graphs in ${{\cal C}}$  is \FPT parameterized by $k$ and $\tau^*$. Our algorithm uses Courcelle's Theorem~\cite{Courcelle90} as a black box, and therefore we did not focus on optimizing the dependence on $k$ and $\tau^*$ of our algorithm. Note that, by \cref{corol-W1}, the \pname{TSS-time} problem is \NP-hard in graphs with maximum degree $\Delta$ for any fixed $\Delta\geq 4$ and $k=\Theta(\log n)$, even if all thresholds are equal to 2. Since graphs of bounded maximum degree have bounded local treewidth, this implies that, even if $\tau^*$ is bounded by a constant, the dependence on $k$ of an \FPT algorithm cannot be of the form $2^{\Ocal(k)}$ unless $\P = \NP$.
Also, what about the hardness of the \textsc{TSS-time} problem in graphs of bounded local treewidth if $k$ is a constant, and $\tau^*$ may depend on $n$? This would be the ``dual'' scenario of the one discussed above for planar graphs and graphs of bounded maximum degree, that is, $\tau^*$ constant and $k$ depending on $n$.

We presented algorithms in time $\Ocal(n)$ and $\Ocal(n^2)$ to find a target set with maximum activation time in a tree for threshold functions and generalized threshold functions, respectively. Obtaining a linear-time algorithm for the latter problem in trees remain open. Finally, can we obtain polynomial-time algorithms in graph classes other than trees? In particular, what about cactus graphs or cographs? \b{Even cliques do not seem to be  completely trivial}.

\bibliographystyle{plainurl}% the mandatory bibstyle
\bibliography{tss-maxtime}

\end{document}